\documentclass[a4paper,11pt]{article}

\usepackage{amsfonts,amssymb,amsmath,amsthm,dsfont,xfrac,xspace}
\usepackage{fullpage,color}
\usepackage{graphicx}
\graphicspath{{./},{images/}}
\makeatletter
\let\NAT@parse\undefined
\makeatother
\usepackage[sort&compress, numbers]{natbib}
\usepackage{hyperref}
\usepackage[footnotesize]{caption}

\newcommand{\adddis}{\Delta_\oplus}
\newcommand{\muldis}{\Delta_\otimes}
\newcommand{\normpdf}{\cl N}
\newcommand{\normsg}{\cl N_{{\rm sg},\alpha}}
\newcommand{\musg}{\mu_{\rm sg}}
\newcommand{\compl}{{\tt c}}
\newcommand{\transp}{\top}
\newcommand{\bqmap}{{\bs A}}
\newcommand{\qmap}{{A}}
\newcommand{\brmap}{{\bs\Phi^{^{_{\hspace{-.4mm}\xi\hspace{-.8mm}}}}}}
\newcommand{\rmap}{\Phi^{^{_{\!\xi}}\!}}

\DeclareRobustCommand{\amgis}{\text{\reflectbox{$\Sigma$}}}
\newcommand{\rv}{\mbox{r.v.}\xspace}
\newcommand{\Qdist}{\cl D}
\newcommand{\bbone}{\bb I}
\newcommand{\gone}{{\textstyle (\!\frac{_{\,^2}}{^\pi}\!)^{\sfrac{^{\!1}\!}{_2}}}}

\newcommand{\kappasg}{\kappa_{_{^{\rm sg}}}}
\newcommand{\Rbb}{\mathbb{R}}
\newcommand{\scp}[2]{\langle #1, #2 \rangle}

\newcommand{\Zbb}{\mathbb{Z}}
\newtheorem{definition}{Definition}
\newtheorem{proposition}{Proposition}

\newtheorem{lemma}{Lemma}

\newcommand{\inv}[1]{\frac{1}{#1}}
\newcommand{\supp}{{\rm supp}\,}
\newcommand{\tinv}[1]{{\textstyle\frac{1}{#1}}}
\newcommand{\sign}{{\rm sign}\,}

\newcommand{\ud}{\mathrm{d}} 

\renewcommand{\leq}{\leqslant}
\renewcommand{\geq}{\geqslant}

\newcommand{\iid}{\mbox{i.i.d.}\xspace}

\DeclareMathOperator{\Id}{\mathds{1}}

\DeclareMathOperator*{\argmin}{argmin}

\newcommand{\bb}{\mathbb}

\newcommand{\bs}{\boldsymbol}
\newcommand{\cl}{\mathcal}
\newcommand{\ie}{\emph{i.e.}, }
\newcommand{\eg}{\emph{e.g.}, }

\title{Small Width, Low Distortions:\\ Quantized Random Embeddings
of  Low-complexity Sets} 

\author{Laurent Jacques\thanks{Image and Signal
    Processing Group (ISPGroup), ELEN Department, ICTEAM institute, Universit\'e catholique de Louvain
(UCL), Belgium. The author is funded by Belgian National Science Foundation
(F.R.S.-FNRS)}}

\begin{document}
\maketitle

\begin{abstract}
Under which conditions and with which distortions can we preserve the pairwise-distances of
  low-complexity vectors, {\em e.g.}, for \emph{structured sets} such
  as the set of sparse
  vectors or the one of low-rank matrices, when these are mapped (or embedded) in a finite set of vectors? 

This work addresses this general question through the specific use of a
quantized and dithered random linear mapping which combines, in the following order, a sub-Gaussian
random projection in $\mathbb R^M$ of vectors in $\mathbb R^N$, a 
random translation, or \emph{dither}, of the projected vectors and
éa uniform scalar quantizer of resolution $\delta>0$ applied
componentwise. 

Thanks to this quantized mapping we are first able to show that, with
high probability, an embedding of a bounded set $\mathcal K \subset
\mathbb R^N$ in $\delta \mathbb Z^M$ can be achieved when distances in the quantized
and in the original domains are measured with the $\ell_1$- and
$\ell_2$-norm, respectively, and provided the number of quantized
observations $M$ is large before the square of the ``Gaussian mean
width'' of $\mathcal K$. In this case, we show that the embedding is actually {\em quasi-isometric} and only suffers of both
multiplicative and additive distortions whose magnitudes decrease as 
$M^{-1/5}$ for general sets, and as $M^{-1/2}$ for structured set,
when $M$ increases. 
Second, when one is only interested in characterizing the maximal
distance separating two elements of $\mathcal K$ mapped to the same
quantized vector, {\em i.e.}, the ``consistency
width'' of the mapping, we show that for a similar number of
measurements and with high probability this width decays as
$M^{-1/4}$ for general sets and as $1/M$ for structured ones when $M$
increases.   
Finally, as an important aspect of our work, we also establish how the non-Gaussianity of
sub-Gaussian random projections inserted in the quantized mapping ({\em e.g.}, for Bernoulli
random matrices)
impacts the class of vectors that can
be embedded or whose consistency width provably decays when $M$
increases.
\end{abstract}

\section{Introduction}
\label{sec:introduction}

There exists an ever-growing trend in high (or ``big'') dimensional data
processing to design new procedures (or to simplify existing ones) using
linear dimensionality reduction (LDR) methods in order to get
faster or memory-efficient algorithms. Provided this reduction does not
bring too much distortion between the initial data space and the
``reduced'' domain, as often allowed by the intrinsic ``low-dimensionality''
properties of the input data, many techniques, such as nearest-neighbor search in big databases \cite{achlioptas2003database,AndDatImm::2006::Locality-sensitive-hashing}, classification
\cite{bandeira2014compressive}, regression \cite{Maillard_JMLR2012},
filtering \cite{davenport2010signal}, manifold processing
\cite{baraniukwakin2009RPmanif} or
compressed sensing \cite{candes2006near,donoho2006cs} can be developed in this reduced domain with
controlled loss of accuracy, as well as stability with respect to data
corruption (\eg noise).         

Most often, those LDR tools rely on defining a random projection
matrix (sometimes called \emph{sensing} matrix) with fewer rows $M$ than columns $N$, whose multiplication with data
represented as a set of vectors in $\bb R^N$ provides a reduced representation (or
\emph{sketch}) of the latter. This is the scheme implicitly promoted for instance by
the celebrated Johnson-Lindenstrauss (JL) lemma for finite sets of vectors $\cl S
\subset \bb R^N$, \ie with $|\cl S| < \infty$
\cite{johnson1984extensions}. This
cornerstone result and its subsequent developments \cite{achlioptas2003database,dasgupta99elementary} showed that,
given a resolution $\epsilon > 0$, if $M \geq C \epsilon^{-2} \log S$
where $S=|\cl S|$ is the
cardinality of $\cl S$ and $C>0$ is a general constant, then a 
random matrix $\bs \Phi \in \bb R^{M\times N}$ whose entries are independently and identically distributed
(\iid) as a centered sub-Gaussian distribution with unit variance
defines an isometric mapping that preserves pairwise-distances between
points in $\cl S$ up to a 
multiplicative distortion $\epsilon$. In other words, $\bs \Phi$ defines
an $\epsilon$\emph{-isometry} between $(\cl S, \ell_2)$ and $(\bs\Phi \cl S,
\ell_2)$, \ie with high probability, for all $\bs x,\bs y \in \cl S$,
\begin{equation}
  \label{eq:JL-ineq}
(1-\epsilon) \|\bs x - \bs y\| \leq \,\tinv{\sqrt M}\|\bs\Phi\bs x - \bs\Phi\bs y\| \leq (1+\epsilon) \|\bs x - \bs y\|.  
\end{equation}
Equivalently, one observes that keeping the probability of success
constant with respect to the random generation of $\bs \Phi$ and
inverting the requirement linking $M$ and $\epsilon$, such an
isometry has a distortion $\epsilon$ decaying as $1/\sqrt M$ when $M$ increases,
\ie this distortion vanishes when $M/\log S$ is large. 
Notice that variants of this embedding result exist with different
``input/output'' norms; see, \eg \cite{krahmer2014unified} for a unified
treatment over a family of \emph{interpolation} norms including $\ell_2$ and
$\ell_1$ as special cases. 

The JL lemma has been later generalized to any subsets $\cl
K \subset \bb R^N$, not only finite, whose typical ``dimension'' can be considered as
small with respect to $N$ (see, \eg
\cite{mendelson2008uniform,baraniukwakin2009RPmanif,dirksen2014dimensionality}). In other
words, as soon as $\cl K$ displays some internal structure that makes
it somehow parametrisable with much fewer parameters than $N$, as for the
set of sparse or compressible signals, the set of low-rank
matrices, signal
manifolds, or a set given as a union of low-dimensional subspaces, an
$\epsilon$-isometry like~\eqref{eq:JL-ineq} can be defined for all
pairs of vectors in $\cl K$.  This is for instance the essence of
the restricted isometry property (RIP) and its link with the JL lemma, where \eqref{eq:JL-ineq} holds
with high probability for all $K$-sparse vectors provided $M \geq C K
\log N/K$ \cite{candes2006near,baraniuk2008simple}.

However, these embeddings have one strong limitation. Except in very specific situations, such as for discrete
sub-Gaussian random matrices $\bs \Phi$ (\eg Bernoulli) and finite
sets $\cl K$,
the set $\bs\Phi \cl K \subset \bb R^M$ is not finite. An infinite number of bits is thus required if one needs to
store, process or transmit $\bs\Phi \bs x$ without information loss for any possible $\bs
x \in \cl K$. Moreover, knowing how many bits are required to
represent such projections is also important theoretically for assessing and measuring the level of information contained in the reduced
data space or for improving specific data retrieval and processing
algorithms. Additionally, if this measure of information can be
achieved, nothing prevents us to take $M \geq N$, as the sought
``dimensionality reduction'' can be aimed at minimizing the number
of bits rather than the dimensionality $M$.   
For instance, \cite{AndDatImm::2006::Locality-sensitive-hashing}
defines locality-sensitive hashing (LSH) as a procedure to turn data
vectors into quantized \emph{hashes} that preserve locality, so that close
vectors induce, with high probability, close hashes. However, this
method is specifically designed for boosting nearest-neighbor searches
over a finite set of vectors and not to define an isometry similar to~\eqref{eq:JL-ineq}.    

As a more practical solution, the embedding realized by a random projection
$\bs \Phi$
is often followed by a scalar quantization procedure, \eg with a
uniform scalar
quantizer $\cl Q: \bb R \to \delta \bb Z$ with resolution $\delta > 0$, applied
componentwise on the image of $\bs \Phi$. A direct impact of this
sequence of operations is to induce a new additive
distortion in~\eqref{eq:JL-ineq} related to $\delta$, as discussed in \cite{BR_DCC13}. 
Indeed, assuming $\bs \Phi$
respects~\eqref{eq:JL-ineq} for all $\bs x$ and $\bs y$ in a
certain subset $\cl K \subset \bb R^N$, given a
uniform quantizer~$\cl
Q(\cdot) := \delta \lfloor \tfrac{\cdot}{\delta} + \tfrac{1}{2} \rfloor$ of resolution $\delta > 0$ applied componentwise on vectors of $\bb
R^M$ we would have $|\cl Q(\lambda) - \lambda| \leq \delta/2$ for all $\lambda \in
\bb R$, which involves $\|\cl Q(\bs u) - \bs u\| \leq \sqrt M \delta /2$ for
any $\bs u \in \bb R^M$. Therefore, a simple manipulation of
\eqref{eq:JL-ineq} provides
\begin{equation}
  \label{eq:QJL-ineq}
(1-\epsilon) \|\bs x - \bs y\| - \delta\ \leq \,\tinv{\sqrt M}\|\cl Q(\bs\Phi\bs
x) - \cl Q(\bs\Phi\bs y)\| \leq\ (1+\epsilon) \|\bs x - \bs y\| + \delta.  
\end{equation}
In other words, as described in Sec.~\ref{sec:framework}, the
quantized mapping $\bs A(\cdot) := \cl Q(\bs \Phi\,\cdot)$ defines now a
\emph{quasi-isometric} embedding between $(\cl K \subset \bb R^N, \ell_2)$ and $(\bs A(\cl
K) \subset \delta \bb Z^M, \ell_2)$.

However, while \eqref{eq:QJL-ineq} displays a constant additive
distortion, several works in this context have observed that
such an additive error actually decays as $M$ increases. First, when distances in the reduced space are
measured with the $\ell_1$-norm and when $\cl Q$ is combined
with a \emph{dithering}\footnote{That is, when the quantizer input is
   randomly shifted \emph{inside} the quantization bin by a random translation adjusted to the quantizer
  resolution \cite{Gray98} (see Sec.~\ref{sec:framework} and~Eq.~\eqref{eq:psi-def}).}, a quasi-isometry similar to
\eqref{eq:QJL-ineq} holds with high probability for all vectors in a
\emph{finite} set $\cl K=\cl S$ \cite{jacques2013quantized} . The additive
distortion reads then $c \delta \epsilon$ for some
absolute constant $c>0$ and this error also decays as $1/\sqrt{M}$, as does the
multiplicative error $\epsilon$.  Second, when combined with universal
quantization \cite{BR_DCC13}, \ie with a periodic scalar quantizer $\cl Q$, 
an exponential decay of this distortion as
$M$ grows can be reached; for the moment, this has been proved only for sparse signal
sets. Finally, recent works related to 1-bit compressed sensing (CS)
have
shown that for a quantization $\cl Q$ reduced to a sign operator (\ie
$\cl Q(\bs \Phi\,\cdot) = \sign(\bs \Phi\,\cdot)$) the angular distance
between any pair of vectors of a low-dimensionality set $\cl K$ is
close to the Hamming distance of their mappings up to an additive
error decaying as $1/M^{1/q}$ for some $q\geq 2$. This is true for random
Gaussian matrices and for the set of sparse signals
\cite{jacques2013robust,plan2011dimension}, for any sets with ``low
dimensionality'' as measured by their Gaussian mean width
\cite{plan2013one,plan2011dimension} (see below) and even for sub-Gaussian random
matrices provided the projected vectors are not ``too sparse''
\cite{ai2014one}, \ie for vectors whose $\ell_\infty$-norm is much
  smaller than their $\ell_2$-norm.    

\noindent \textbf{Contributions:} Considering
these last observations, the main results of this paper
show that:\\[-7mm]
\begin{itemize}
\item[\emph{(i)}] quasi-isometric embeddings can be obtained with high probability from scalar (dithered) quantization after
linear random projection; for such embeddings both multiplicative and additive
distortions co-exist when, as in \cite{jacques2013quantized},
distances between mapped vectors are
measured with the $\ell_1$-norm\footnote{Notice that for
      binary embeddings the Hamming distance separating the binary
      mapping of two vectors, as used in \cite{Jacques2010,plan2011dimension}, is also
  the half of their $\ell_1$-distance.}; 
\item[\emph{(ii)}] random sensing matrices for such embeddings are
allowed to be generated from symmetric sub-Gaussian distributions provided
embedded vector differences are not ``too sparse'' (as in the 1-bit
case~\cite{ai2014one}); 
\item [\emph{(iii)}] the 
results above actually hold with high probability for \emph{any} subset $\cl K$ of $\bb
R^N$ as soon as $M$ is large compared to its typical dimension, \ie to
its squared 
Gaussian mean width.
\item [\emph{(iv)}] with high probability, the biggest distance separating two
  \emph{consistent} vectors in $\cl K$ (\ie characterized by identical
quantized mappings), that is what we call the \emph{consistency
  width}, decays when $M$ increases at a faster rate than what could be
predicted by using just the implications of a quasi-isometry. This extends
to any set $\cl K$ the works of
\cite{powell_consistent,jacques2014error}, that were valid only for sparse signals;
\item [\emph{(v)}] for particular \emph{structured} sets, \eg the
  set of (bounded) sparse vectors or the set of (bounded) low-rank matrices, 
  the minimal values of $M$ necessary to specify a quantized embedding
  or a small consistency width can be strongly reduced compared to those
  required for a general set;
\end{itemize}
Moreover, we aim at optimizing whenever it is possible the
requirements on $M$ (\eg with respect to $\epsilon$ and $\delta$) that
guarantee those results. 
\medskip

\noindent \textbf{Methodology:} As an important aspect of our developments, we study the conditions
for obtaining quasi-isometric embeddings of
any bounded subsets $\cl K \subset \bb R^N$ into $\delta \bb Z^M$. Following
key procedures established in other works
\cite{plan2011dimension,plan2012robust}, the typical dimension of these sets
is measured by the \emph{Gaussian mean width}, \ie  
$$
w(\cl K):=\bb E \sup_{\bs u \in \cl K} |\bs g^\transp \bs u|,
$$ 
with $\bs g
\sim \normpdf^N(0,1)$. This quantity, also known as Gaussian complexity,
has been recognized as central for instance in characterizing random
processes \cite{vaart1996weak}, shrinkage estimators in signal
denoising and high-dimensional statistics
\cite{chandrasekaran2012computational}, linear inverse problem solving
with convex optimization \cite{chandrasekaran2012convex} or
classification efficiency for randomly projected signal sets
\cite{bandeira2014compressive}. More specifically, the
minimal number of measurements $M$ necessary to induce, with high
probability, an $\ell_2/\ell_2$-isometric embedding of any subset $\cl
K \subset \bb S^{N-1}$
into $\bb R^M$ from sub-Gaussian random projections is known to be
proportional to $w(\cl K)^2$
\cite{mendelson2008uniform}. Therefore, since $w(\cl K)^2 \lesssim
\log |\cl K|$ for some finite set $\cl K$, we recover the condition
defining the Johnson-Lindenstrauss lemma by imposing $M \gtrsim \log
|\cl K|$ \cite{johnson1984extensions},
while for the set of bounded
$K$-sparse vectors in an
orthonormal basis (ONB) $\bs \Psi \in \bb R^{N\times N}$, $w(\cl K)^2 \lesssim K \log
N/K$, which characterizes the conditions of the restricted isometry
property (RIP) for sub-Gaussian random matrices
\cite{baraniuk2008simple}. The interested reader can find a summary of the main properties of the
Gaussian mean width in Table~\ref{tab:Gaussian-mean-width-prop}, with
explicit references to their origin. This table
could be helpful also to keep trace of these properties while reading our proofs.

In our developments, we sometimes complete the characterization of
sets provided by the Gaussian mean width with another important
measure: the Kolmogorov $\epsilon$-entropy of a set
$\cl K \subset \bb R^N$ that we denote $\cl H(\cl K,\epsilon)$
\cite{kolmogorov1959varepsilon}. This is defined as the logarithm
of the size of the smallest $\epsilon$-net of $\cl K$, \ie a set $\cl
C_\epsilon(\cl K) \subset \cl K$ such
that any vector of $\cl K$ cannot be farther than $\epsilon$ from its
closest vector in $\cl
C_\epsilon(\cl K)$. By the
Sudakov inequality, this entropy is connected to the Gaussian mean
width as $\cl H(\cl K,\epsilon) \leq w(\cl K)^2/\epsilon^2$. 

However, in specific cases this last inequality
is too loose with respect to $\epsilon$. As summarized in
  \cite{oymak2015nearoptbin}, this is the case of the \emph{structured
    sets} $\cl K$ defined hereafter, for which this work will provide separated and
  tighter results.
  \begin{definition}[Structured sets\footnote{Notice that in
          \cite{oymak2015nearoptbin} $\cl K$
          is assumed to be a subset of the sphere $\bb S^{N-1}$ so
          that $d=1$. However, this slight difference
          does not change the bound on the Kolmogorov entropy or
          the Gaussian mean width of the structured sets
  considered in \cite{oymak2015nearoptbin} and in this paper.  
        } \cite{oymak2015nearoptbin}] 
\label{def:structured-set}
A bounded set $\cl K \subset \bb R^N$ with diameter $d = \|\cl K\| :=
\max\{\|\bs u\|:\bs u \in \cl K\} < \infty$ is \emph{structured} iff
there exists a quantity $\bar w(\cl K)$, independent of $d$, 
    for which we have both
    \begin{subequations}
      \label{eq:structured-set-prop}
      \begin{align}
      \label{eq:structured-set-prop-kolmo}
      \textstyle \cl H(\cl K,\epsilon)&\textstyle \leq \bar w(\cl K)^2 \log(1 +
      \frac{d}{\epsilon}),\\
      \label{eq:structured-set-prop-local}
      w(d^{-1}\,\cl K_{\epsilon d})^2&\textstyle = w\big(({d^{-1}\cl K} - {d^{-1}\cl K}) \cap
    \epsilon \bb B^n\big)^2\leq \epsilon^2\,\bar w(\cl K)^2,
    \end{align}
  \end{subequations}
  for any $\epsilon>0$, where $\cl K_{\epsilon'} := ({\cl K} - {\cl K}) \cap
  \epsilon' \bb B^n$ is the \emph{local set} of $\cl K$ of radius
  $\epsilon' > 0$. 
  \end{definition}
For instance, if $\cl K'$ is a subspace of $\bb R^N$, a union
of subspaces (such as the
set $\Sigma^{\bs \Psi}_K$ of $K$-sparse signals in an orthonormal basis or in a
redundant dictionary $\bs \Psi$ of $\bb R^N$), the set of rank-$r$
matrices $\cl M_r$ in $\bb R^{N_1 \times N_2}$, or even
the set of group-sparse signals, then $\cl
K'$ is a \emph{cone}, \ie $\lambda \cl K' \subset \cl
K'$ for any $\lambda >0$, and the
set $\cl
K := \cl K' \cap d\,\bb B^N$ is
structured for any diameter $d>0$ \cite{oymak2015nearoptbin}. 

Indeed, 
focusing first on \eqref{eq:structured-set-prop-local}, 
if $\cl K'$ is one of the sets listed above, $\cl K'' := \cl K' - \cl
K'$ is also a cone and $\cl K'' \supset d^{-1}(\cl K -
\cl K)$. Therefore $w(d^{-1}\cl K_{\epsilon d})^2 \leq w( \cl K'' \cap \epsilon \bb B^N)^2 =
\epsilon^2  w(\cl K'' \cap \bb B^N)^2$. This last quantity is easily bounded since $\cl K''$ often shares the same structure than $\cl
K'$, \eg $\cl K'' = \Sigma^{\bs \Psi}_{2K}$ if $\cl K'=\Sigma^{\bs
  \Psi}_K$, and in fact $w(\cl K'' \cap \bb
B^N) \simeq w(\cl K/\|\cl K\|)$ showing that $\bar w(\cl K)$ can be set
to $w(\cl K/\|\cl K\|)$ in~\eqref{eq:structured-set-prop-local}.

Second, for 
\eqref{eq:structured-set-prop-kolmo}, the Komogorov entropy of such a set
$\cl K'$ can often be tightly bounded by 
decomposing it into a union
of subspaces or subdomains restricted to $d\,\bb B^N$, so that a global
$\epsilon$-net of small cardinality could be reached by
the union of the $\epsilon$-nets of all of these subparts~\cite{oymak2015nearoptbin,baraniuk2008simple,pisier}, \ie justifying the bound $\cl H(\cl K' \cap d
\bb B^N) \leq \bar w(\cl K)^2 \log(1 + \frac{d}{\epsilon})$.
Actually, concerning \eqref{eq:structured-set-prop-kolmo}, it occurs
that for all the structured sets listed above we have that either $\bar w(\cl K)^2 
\simeq w(\cl K/\|\cl K\|)^2$ or both $\bar w(\cl K)^2$ and $w(\cl K/\|\cl
K\|)^2$ have the same simplified closed-form upper bound,
\eg they are both upper bounded by $K
\log(N/K)$ when~$\cl K'=\Sigma^{\bs \Psi}_K$. 

Thus, due to the observations made above, we will consider that $\bar
w(\cl K)$ can be bounded similarly to the actual Gaussian mean width
$w(\|\cl K\|^{-1}\,\cl K)$ of the normalized set $\|\cl K\|^{-1}\,\cl K$, \ie with the
same simplified upper bound. An example of this fact for the set of bounded
$K$-sparse vectors is provided at the end of Sec.~\ref{sec:proofs}.

\begin{table}[!t]
  \newcounter{mycount}
  \setcounter{mycount}{0} 
  \newcommand{\myenum}
  {\refstepcounter{mycount}{\scriptsize\em (P\arabic{mycount})}~}
  \centering
  \footnotesize
  \begin{tabular}{|@{\,}l@{\,}p{.285\textwidth}@{\,}|p{.6\textwidth}@{\ }|}
    \hline
    &\bf Names&\bf Properties\\
    \hline
    \myenum\label{p:w-def}&Definition&$w(\cl A) = \bb E \sup_{\bs x \in \cl A} |\scp{\bs
                g}{\bs x}|$ for $\bs g \sim \normpdf^N(0, 1)$.\\
    \myenum\label{p:w-hom}&Homogeneity~\scriptsize{\cite[Sec. 3.2]{chandrasekaran2012convex}}&$w(\lambda A) = \lambda w(\cl A)$ for $\lambda > 0$.\\ 
    \myenum\label{p:w-setinc}&Set inclusion~\scriptsize{\cite[Sec. 3.2]{chandrasekaran2012convex}}&if $\cl A\subset \cl B$, $w(\cl A) \leq w(\cl B)$.\\
    \myenum\label{p:w-setdiff}&Set difference~\scriptsize{\cite[Sec. 5.3]{plan2011dimension}}&$w(\cl A - \cl A) \leq 2 w(\cl A)$.\\
    \myenum\label{p:w-modul}&Modularity~\scriptsize{\cite[Sec. 3.2]{chandrasekaran2012convex}}&$w(\cl A \cup \cl B) + w(\cl A \cap \cl B) = w(\cl A)
                 + w(\cl B)$, if $\cl A$, $\cl B$ and $\cl A \cup \cl
                                                                                            B$
                                                                                            are
                                                                                            convex.\\
    \myenum\label{p:w-chull}&Convex hull~\scriptsize{\cite[Sec. 3.2]{chandrasekaran2012convex}}&$w({\rm conv}(\cl A)) = w(\cl A)$.\\
    \myenum\label{p:w-subsp}&Subspace~\scriptsize{\cite[Sec. 3.2]{chandrasekaran2012convex}}&if $\cl A_K$ is a $K$-dimensional subspace of $\bb R^N$, then\newline
               \phantom{abcd}$w(\cl A_K \cap \bb S^{N-1}) = w(\cl A_K \cap \bb B^{N})
               \leq \sqrt K$.\\
    \myenum\label{p:w-subadd}&Subspace
                               addition~\scriptsize{\cite[Eq. (15)]{chandrasekaran2012convex}}&$w((\cl A_K \oplus \cl B)\cap \bb S^{N-1})^2 \leq
                        K +  w(\cl B \cap \bb S^{N-1})^2$.\\
    \myenum\label{p:w-dialk}&Link with diameter$^*$&for $\|\cl A\| := \sup_{\bs u \in
                        \cl A} \|\bs u\|$,\newline
\phantom{abcd}$\gone\,\|\cl A\| \leq w(\cl A) \leq \sqrt
                        N\,\|\cl A\|$.\\
    \myenum\label{p:w-symm}&Symmetrization$^*$&$w(\cl A) - \gone \inf_{\bs u \in \cl
                                A}\|\bs u\|\, \leq \bb E \sup_{\bs x \in \cl A - \cl A} |\scp{\bs
                g}{\bs x}| \leq 2 w(\cl A)$.\\
    \myenum\label{p:w-trans}&Translation$^*$&$w(\cl A) -\gone\,\|\bs t\|\,\leq w(\cl A+\{\bs t\})\leq w(\cl A) +\gone\,\|\bs t\|$, for $\bs
                  t\in \bb R^N$.\\
    \myenum\label{p:w-invort}&Invariance under $\cl O_N$\newline\scriptsize{\cite[Prop. 2.1]{plan2012robust}}&For all $\bs B \in \cl O_N := \{\bs C \in \bb R^{N\times
                             N}: \bs C\bs C^\transp = \bs C^\transp\!\bs C = \Id_N\}$,\newline
                                         $w(\bs B\,\cl A) = w(\cl A)$.\\ 
    \myenum\label{p:w-transor}&Translation on origin
                                \newline\scriptsize{(from
                                {(P\ref{p:w-dialk})} \& {(P\ref{p:w-trans})}})&$w(\cl A) - \gone \|\bs x_0\| \leq w(\cl A
                            - \{\bs x_0\}) \leq 2 w(\cl A)$ \newline
for $\bs x_0 \in \cl A$ with $\|\bs x_0\| \geq \inf_{\bs u \in \cl
                                A}\|\bs u\|$.\\
    \myenum\label{p:w-suda}&Sudakov
                             inequality~\scriptsize{\cite[Sec. 1.7]{plan2011dimension}}&For an $\epsilon$-net $\cl G_\epsilon \subset
                        \cl A$, $\log |\cl G_\epsilon| \lesssim
                        \epsilon^{-2}\, w(\cl A)^2$.\\  
\hline
\multicolumn{3}{c}{}\\[-2mm]
\hline
&\bf Special sets&\bf Widths\\
\hline
    \myenum\label{p:w-finset}&Finite~\scriptsize{\cite[Sec. 1.4]{plan2011dimension}}&$w(\cl S)^2 \lesssim \log |\cl S|$.\\
    \myenum\label{p:w-sph}&Sphere and ball~\scriptsize{\cite[Sec. 1.4]{plan2011dimension}}&$w(\bb S^{N-1}) \leq \sqrt N$ and $w(\bb
                             B^{N}) \leq \sqrt N$.\\
    \myenum\label{p:w-spars}&Sparse signals~\scriptsize{\cite[Sec. 1.3]{plan2011dimension}}&For $\Sigma_K := \{\bs u: \|\bs u\|_0 \leq \sqrt K\}$, $w(\Sigma_K\cap\,\bb B^N)^2 \lesssim K \log(2N/K)$.\\
    \myenum\label{p:w-compr}&``Compressible signals''\newline\scriptsize{\cite[Sec. 1.3]{plan2011dimension}}&For $\cl K_{N,K} := \{\bs u: \|\bs u\|_1 \leq \sqrt K,
                               \|\bs u\|\leq 1\}$,\newline \phantom{abcd}$w(\cl K_{N,K})^2
                                                       \lesssim K
                                                       \log(2N/K)$.\\ 
    \myenum\label{p:w-lowrank}&Low-rank
matrices\newline\scriptsize{\cite[Lemma 21]{kabanava15}}&For $\cl M_{r}
 := \{\bs U \in \bb R^{N_1 \times N_1}: {\rm rank}(\bs U) \leq r\}$,\newline
\phantom{abcd}$w(\cl M_{r})^2 \lesssim r\,(N_1 + N_2)$.\\
    \hline
  \end{tabular}
  \caption{Useful properties of the Gaussian mean width. If not
    otherwise noted, all sets are
    subsets of~$\bb R^N$. $*$: {(P\ref{p:w-dialk})} is obtained by a
      simple use of the Jensen and Cauchy-Schwartz
      inequalities, {(P\ref{p:w-trans})} is a simple consequence of the
      triangular inequality and of $\bb E
      |\scp{\bs g}{\bs t}| =  \gone\,\|\bs t\|$.}
  \label{tab:Gaussian-mean-width-prop}
\end{table}
\medskip

\noindent \textbf{Paper organization:} The rest of the paper is structured as follows. In
Sec.~\ref{sec:framework}, we define the construction of our quantized
sub-Gaussian random mapping. Additionally, this section characterizes the sub-Gaussianity of
its linear ingredient, \ie its random projection matrix,
and its interplay with the ``anti-sparse'' nature of the mapped
vectors. We also formalize and motivate the main objectives of the
paper, \eg explaining the shape and the origins of the targeted quasi-isometric
embedding with its two specific
distortions. Sec.~\ref{sec:main-results} provides the main results of
this work, namely, \emph{(i)} the possibility to create with high probability a
quasi-isometric sub-Gaussian embedding from our quantized
mapping (Prop.~\ref{prop:main-result}), and \emph{(ii)} a study of
this mapping's \emph{consistency width} behavior 
(Prop.~\ref{prop:consistency-width}). Sec.~\ref{sec:discussions}
discusses those two propositions, analyzing them in a few specific
settings in comparison with related works in the fields of dimensionality reduction
and 1-bit compressed sensing. Sec.~\ref{sec:nec-dither} questions the
necessity of dithering in the mapping $\bqmap$ and shows that, from an
appropriate counterexample, our results
do not hold in full generality without such a dither. Finally, Sec.~\ref{sec:proofs} and
Sec.~\ref{sec:proof-consistency-width} contain the proofs of
Prop.~\ref{prop:main-result} and Prop.~\ref{prop:consistency-width},
respectively, the auxiliary Lemmas being demonstrated in appendix.         

\paragraph*{Conventions:} We find useful to summarize here our mathematical notations. Domain dimensions are denoted
by capital roman letters, \eg $M, N, \ldots$ Vectors and matrices are associated to bold
symbols, \eg $\bs \Phi \in \bb R^{M\times N}$ or $\bs u \in \bb
R^M$, while lowercase light letters are associated to scalar
values. The identity matrix in $\bb R^D$ reads $\Id_{D}$ while
$\bbone[A]\in\{0,1\}$ is the indicator function of a set $A \subset
\bb R^D$. An ``event'' is a set
whose definition depends on the realization of some random variables,
\eg if $X\in \bb R$ is a random variable, the event $A=\{X \leq 0\}$
has probability $\bb P(X\leq 0) = \bb E\,\bbone[A]$. The $i^{\rm th}$
component of a vector (or of a vector function) $\bs u$ reads either $u_i$ or $(\bs u)_i$,
and the vector $\bs u_i$ may refer to the $i^{\rm
  th}$ element of a set of vectors. The set of indices in
$\Rbb^D$ is $[D]=\{1,\,\cdots,D\}$. The cardinality of a finite set $\cl J$ reads $|\cl J|$.
For any $p\geq 1$, the $\ell_p$-norm of $\bs u$ is $\|\bs u\|_p^p = \sum_i |u_i|^p$ with
$\|\!\cdot\!\|:=\|\!\cdot\!\|_2$. The ``$\ell_0$-norm'' of a vector $\bs
u \in \bb R^N$ is $\|\bs u\|_0 = |\supp \bs u|$, with $\supp \bs u =
\{i: u_i \neq 0\}$ the support of $\bs u$. The $(N-1)$-sphere in $\Rbb^N$ is $\bb S^{N-1}=\{\bs x\in\Rbb^N:
\|\bs x\|=1\}$ while the unit ball is denoted $\bb B^{N}=\{\bs x\in\Rbb^N:
\|\bs x\|\leq 1\}$. The diameter of a bounded set $\cl A \subset \bb R^N$ is written $\|\cl
A\| = \sup\{\|\bs u\|: \bs u \in \cl A\}$. The set of $K$-sparse signals in $\bb R^N$
is defined as $\Sigma_K := \{\bs u \in \bb R^N: \|\bs u\|_0 \leq K\}$
while the set of $K$-sparse signals in an orthonormal basis (ONB) $\bs \Psi
\in \bb R^{N \times N}$, \ie with $\bs\Psi \bs\Psi^\transp=\bs\Psi^\transp \bs\Psi =
\Id_N$, reads $\Sigma^{\bs \Psi}_K = \bs\Psi\Sigma_K$. 
The positive thresholding function is defined by
$(\lambda)_+ := \tinv{2}(\lambda + |\lambda|)$ for
any~$\lambda\in\Rbb$. For $t \in \bb R$, $\lfloor t \rfloor$
(resp. $\lceil t \rceil$) is the
largest (smallest) integer smaller (greater) than $t$.  
A random matrix $\bs \Phi \sim \cl P^{M\times N}(\Theta)$ is a $M\times N$
matrix with entries distributed as $\Phi_{ij} \sim_{\iid} \cl
P(\Theta)$ given the distribution parameters $\Theta$ of $\cl P$ (\eg $\cl
N^{M\times N}(0,1)$ or $\cl U^{M\times N}([0,1])$). A
random vector in $\Rbb^M$ following $\cl P(\Theta)$ is defined by $\bs v \sim \cl P^{M}(\Theta)$.
Given two random variables $X$ and $Y$, the notation $X \sim
Y$ means that $X$ and $Y$ have the same distribution. Since our
developments do not focus on sharp bounds, we denote by $C,c,c'$ or $c''$
(possibly large) constants whose value can change between
lines. In a few places, for simplicity, we write $f\lesssim g$ if there exists a constant $c>0$ such
that $f \leq c\,g$, and correspondingly for $f\gtrsim g$. Moreover,
$f\simeq g$ means that $f\lesssim g$ and $g\lesssim f$.  Finally, for
asymptotic relations, we use the common Landau
family of notations, \ie the symbols $O$, $\Omega$ and
$\Theta$~\cite{knuth1976big}.

\section{Quantized Sub-Gaussian Random Mapping}
\label{sec:framework}

In this work, given a quantization resolution $\delta>0$, we
focus on the interaction between a random projection of $\bb R^N$ into
$\bb R^M$ and the following uniform (dithered)
quantizer\footnote{Hereafter, our developments could be
      adapted to any quantizer defined as $\cl Q'(t) := \delta (\lfloor \tfrac{t + q_0}{\delta}\rfloor + r_0)
\in \delta\bb Z$, for some $q_0 \in [0, \delta)$ and $r_0 \in
[0,1)$, \eg for the quantizer mentioned in the Introduction with $r_0 = 0$ and $q_0 = \delta/2$.}
$\cl Q(t) = \delta \lfloor \tfrac{t}{\delta}
\rfloor \in \delta\bb Z$, applied componentwise on vectors in $\bb R^M$. In other words, for some random matrix $\bs \Phi \in \bb R^{M \times N}$ whose distribution
is specified below, we study the properties of the 
mapping $\bqmap:\bb R^N \to \delta\Zbb^M$ with
\begin{equation}
  \label{eq:psi-def}
  \bqmap(\bs x) := \cl Q(\bs\Phi\bs x + \bs \xi),  
\end{equation}
where $\bs \xi \in \cl U^M([0, \delta])$  is a uniform
\emph{dithering} that stabilizes 
the action of $\cl Q$
\cite{B_TIT_12,Gray98,jacques2013quantized}.

We specialize the mapping \eqref{eq:psi-def} on projection (or sensing) matrices $\bs \Phi$ with
entries independently and identically drawn from a
symmetric \emph{sub-Gaussian} distribution. We recall that a random variable (\rv) $X$ is
sub-Gaussian if its \emph{sub-Gaussian norm} (or $\psi_2$-norm) \cite{IntroNonAsRandom}
\begin{equation}
  \label{eq:sub-gaussian-def}
  \|X\|_{\psi_2}\ :=\ \sup_{p\geq 1}\ p^{-\sfrac{1}{2}} (\bb E |X|^p)^{\sfrac{1}{p}}.  
\end{equation}
is finite\footnote{Notice
  that other equivalent definitions for sub-Gaussian \rv exist, see
  \eg \cite{mendelson2008uniform}.}. Examples of
sub-Gaussian \rv 's are Gaussian, Bernoulli, uniform or bounded
\rv's, as 
$$
\|X\|_{\psi_2}\ \leq\ \|X\|_{\infty} := \inf\{t\geq 0: \bb P( |X| \leq
t) = 1\}.
$$
Sub-Gaussian \rv's are endowed with several interesting properties
described, \eg in \cite{IntroNonAsRandom}. Their tail
is for instance bounded as the one of a Gaussian \rv, \ie there exists a $c>0$
such that for all $\epsilon\geq 0$ and for a sub-Gaussian \rv $X$,
\begin{equation}
  \label{eq:sg-tail-bound}
  \bb P(|X|>\epsilon)\ \lesssim\ e^{-c\,\epsilon^2/ \|X\|^2_{\psi_2}}.  
\end{equation}
Moreover, since $\|X - \bb E X\|_{\psi_2} \leq
\|X\|_{\psi_2} + \|\bb E X\|_{\psi_2} =  \|X\|_{\psi_2} + |E X| \leq
\|X\|_{\psi_2} + E|X| \leq 2\|X\|_{\psi_2}$, centering $X$ has no effect
on its sub-Gaussianity.

By a slight abuse of notation, we denote collectively the distributions of \emph{symmetric} sub-Gaussian
\rv with zero expectation, unit variance and finite sub-Gaussian norm $\alpha$
by $\normsg(0,1)$, with $\alpha \geq 1/\sqrt 2$ from
\eqref{eq:sub-gaussian-def}. This means that if $X \sim \normsg(0,1)$,
we do not fully specify the pdf of $X$ but we know that $X$ is centered, has
unit variance and sub-Gaussian norm $\alpha$.

In this context, for a sub-Gaussian random
matrix $\bs \Phi = (\bs
\varphi_1,\,\cdots,\bs\varphi_M)^\transp \sim
\normsg^{M \times N}(0,1)$, each row $\bs \varphi_i$ is also \emph{isotropic}, \ie for all
$i\in [M]$ and all $\bs u \in \bb R^N$, 
$$
\bb E|\scp{\bs\varphi_i}{\bs
  u}|^2 = \|\bs u\|^2.
$$ 
However, conversely to the Gaussian case where $\bb E|\scp{\bs g}{\bs
  u}| = \gone\|\bs u\|$ for $\bs g \sim \cl N^N(0,1)$ and $\bs u \in
\bb R^N$ (since $\scp{\bs g}{\bs
  u} \sim \cl N(0, \|\bs u\|^2)$), we do not necessarily have $\bb E|\scp{\bs\varphi}{\bs
  u}| = c\|\bs u\|$ for $\bs \varphi \sim
\normsg^N(0,1)$ and some absolute constant $c>0$. 

As will be clear below, we must anyway determine the deviations to this last equality.  
Interestingly, as noted in~\cite{ai2014one}, any sub-Gaussian random
  vector $\bs \varphi \sim
  \normsg^N(0,1)$ satisfies
\begin{equation}
  \label{eq:Berry-Esseen-relation}
  \int_{0}^{+\infty} \big|\bb P(|\scp{\bs \varphi}{\bs u}| \geq t)\ -\
  \bb P(|\scp{\bs g}{\bs u}| \geq t)\big|\,\ud t\ \leq\ \kappasg \|\bs u\|_{\infty},\quad \forall \bs u \in \bb R^N,
\end{equation}
for some constant $\kappasg \geq 0$ depending only the
distribution of $\bs \varphi \sim \normsg^N(0,1)$. While we have
obviously $\kappasg =
0$ if $\bs \varphi \sim \cl N^N(0,1)$, it is possible to bound
this constant in full generality. Indeed, up to a simple change of variable $t\to t
\|\bs u\|$ in the integral, \eqref{eq:Berry-Esseen-relation} is
sustained by the Berry-Esseen central limit theorem 
(as described in a simplified form in \cite[Theorem 4.2]{ai2014one}).
This result shows basically that, for $\bs u \in \bb S^{N-1}$, the LHS
of \eqref{eq:Berry-Esseen-relation} is bounded by $9\,\bb E|\varphi|^3\,\|\bs u\|_3^3 \leq 9\sqrt{27}\,\alpha^3 \|\bs u\|_\infty$ for $\varphi_i \sim_{\rm i.i.d.}
\varphi \sim \normsg(0,1)$. This means
that $\kappasg \leq 9\sqrt{27}\,\alpha^3$ for any
$\bs \varphi \sim \normsg^N(0,1)$. Notice, however, that this bound can be loose
  for many sub-Gaussian distributions.

Thanks to assumption \eqref{eq:Berry-Esseen-relation}, we can establish the
behavior of the \emph{first absolute moment} function
\begin{equation}
  \label{eq:musg-def}
  \musg(\bs u) := \bb E |\scp{\bs \varphi}{\bs u}|.
\end{equation}
Since $\bb E|X| = \int_0^\infty \bb P(|X| \geq t)\, \ud
t$ for any \rv $X$ and using Jensen's inequality, we indeed observe~that
\begin{equation}
  \label{eq:bound-musg}
  \musg(\bs u)\ \leq\ (\bb E |\scp{\bs \varphi}{\bs u}|^2)^{1/2} = \|\bs u\|,
\end{equation}
\begin{equation}
  \label{eq:one-moment-sg-antisparse}
  \big|\,\musg(\bs u)\ -\ \gone \|\bs u\|\,\big|\ \leq\ \kappasg \|\bs u\|_{\infty},
\end{equation} 
for all $\bs u \in \bb R^N$. The last property, which is also considered in 1-bit CS
with non-Gaussian projections \cite{ai2014one}, is key for
characterizing quantized embeddings from sub-Gaussian projections. 
\medskip

Having now fully described the elements composing our random quantized mapping $\bs A$, we formally address the objectives defined in
the Introduction by observing ``when'', \ie under which conditions
with respect to $M$, there exist two small distortions $\adddis, \muldis \geq 0$ such that the pseudo-distance $\Qdist(\bs x,\bs y) := \tinv{M}\,\| \bqmap(\bs x) -
  \bqmap(\bs y)\|_1$ is involved in the quasi-isometric relation
\begin{equation}
  \label{eq:objective-paper}
  \textstyle \big|\,\Qdist(\bs x,\bs y)\ -\ \gone \|\bs x - \bs y\|\,\big|\
  \leq\ \muldis\,\|\bs x - \bs y\|\ +\ \adddis,
\end{equation}
for all pair of vectors taken in a general subset $\cl K \subset \bb R^N$. 

In particular, we aim to control the distortions $\adddis$ and
$\muldis$ with
respect to $M$, $N$, the non-Gaussian nature of $\bs \Phi$ (\ie through
$\alpha$ and $\kappasg$), the typical dimension of $\cl K$ (\ie its
Gaussian mean width) and possible additional requirements on $\bs x$
and $\bs y$. 

Let us justify and comment the specific form taken by
\eqref{eq:objective-paper}. First, $\Qdist$ is associated to a $\ell_1$-distance in the image of
$\bqmap$. As detailed in
Sec.~\ref{sec:proofs}, this choice establishes an equivalence between
the evaluation of $\Qdist$ and a specific counting procedure, \ie a
count of the
number of quantization \emph{thresholds} separating each components of the
randomly-projected vectors. However, it is not
clear if our developments can be extended to a $\ell_2$-based
pseudo-distance, even if this holds, with additional distortion,
in the case of Gaussian random projections and for finite sets $\cl K$
\cite{jacques2013quantized} (see Sec.~\ref{sec:discussions}).    

Second, as explained in the Introduction, a special case where
  both non-zero $\adddis$ and $\muldis$ appear specifies the
constant $\gone$ in \eqref{eq:objective-paper}. When $\bs \Phi \sim \cl N^{M \times N}(0,1)$, \cite{jacques2013quantized} has proved a
quantized version of the Johnson Lindenstrauss (JL) Lemma showing that for
a finite set $\cl S \subset \bb R^N$ of size $S$, provided $M \gtrsim
\epsilon^{-2} \log S$, one has
$$
\textstyle |\Qdist(\bs x,\bs y) - \gone\,\|\bs x
- \bs y\||\ \lesssim\ \epsilon\,\|\bs x - \bs y\| + \epsilon\delta,
$$
for all pairs $\bs x,\bs y \in \cl S$ with a probability at least
$1-e^{-\epsilon^2 M}$. As a direct impact of the loss of
  information induced by the quantization, we also observe here that $\bqmap$ realizes a
\emph{quasi-isometric} mapping between $(\cl S \subset \bb R^N, \ell_2)$ and
$(\bqmap(\cl S) \subset \delta \bb Z^M, \ell_1)$ with $\muldis =
\epsilon$ and $\adddis = \delta\epsilon$. 

Finally, as will be clearly established in
  Sec.~\ref{sec:quasi-isometric-embed}, the anti-sparse nature of $\bs x - \bs y$ must be
  involved in the characterization of the right-hand side of ~\eqref{eq:objective-paper} in the case of a general
  sub-Gaussian matrix $\bs \Phi$. 
Indeed, let us consider a matrix with i.i.d. Bernoulli
  distributed random entries, \ie $\Phi_{ij} \sim_{\rm iid}
\cl B(\tinv{2})$ with $\bb P(\Phi_{ij} = 1)= \bb P(\Phi_{ij} = -1) = 1/2$ for
all $1\leq i\leq M$ and $1\leq j\leq N$, the vectors $\bs x =
(1,0,\cdots,0)^\top \in \bb R^N$ and $\bs y = \bs 0 \in \bb R^N$ and
assume $\bs x, \bs y \in \cl K$, \eg with $\cl K = \Sigma_K \cap \bb
B^N$ and $K \geq 1$. Then, taking $\delta = 1$, we clearly have $\bs A(\bs x) \in \{\pm 1\}^M$ and $\bs A(\bs y)=0$, so that $\cl D(\bs x,\bs y) = 1$ and $\|\bs
x - \bs y\| = 1$. Consequently, if \eqref{eq:objective-paper} is
expected to hold on any pair of vectors in $\cl K$, inserting $\bs x$
and $\bs y$ inside it gives $\adddis +
\muldis \geq 1 - \gone > 0.202$. This limits our hope to have $\adddis +
\muldis$ as small as we want by, \eg increasing $M$.

In fact, between the two distortions, it is actually $\muldis$ that should
  depend on the configuration of $\bs x-\bs y$. As
proved in App.~\ref{sec:absol-expect-dith},
\begin{equation}
  \label{eq:abso-expec-dither-floor}
  \bb E |\lfloor x + \xi \rfloor - \lfloor y + \xi \rfloor| = |x -
  y|,\quad \forall x,y \in \bb R,\ \xi \sim \cl U([0,1]).
\end{equation}
Therefore, by definition of $\cl Q$, from the independence of each
component of $\bs A$ and using the law
of total expectation over $\bs \xi$ and $\bs \Phi$ we have
\begin{equation}
  \label{eq:law-of-total-expectation-on-Qdist}
  \bb E\,
  \Qdist(\bs x,\bs y) = \bb E_{\bs \varphi} \bb E_{\xi} |\cl
  Q(\bs\varphi^\transp \bs x + \xi) - Q(\bs\varphi^\transp \bs y + \xi)|
  = \bb E_{\bs \varphi} |\bs\varphi^\transp(\bs x - \bs y)| =
  \musg(\bs x - \bs y),  
\end{equation}
with $\bs
\varphi \sim \normsg^N(0,1)$ and $\xi \sim \cl U([0,\delta])$. From the assumption (\ref{eq:one-moment-sg-antisparse}) and
given $K_0 \in \bb R$, we then observe that 
\begin{equation}
  \label{eq:expect-qdist-on-antisparse}
\textstyle |\bb E\, \Qdist(\bs x,\bs y) - \gone\,\|\bs x - \bs y\||\
=\ |\,\musg(\bs x- \bs y) - \gone\,\|\bs x - \bs y\||\ \leq\ \tfrac{\kappasg}{\sqrt K_0}\,\|\bs x - \bs y\|,  
\end{equation}
for all vectors $\bs x$ and $\bs y$ such that $\bs x - \bs y$ belongs
to the set\footnote{That could be pronounced ``amgis''.}
\begin{equation}
  \label{eq:anti-sparse-set}
\amgis_{K_0} := \{\bs u \in \bb R^N: K_0\|\bs u\|^2_\infty \leq \|\bs
u\|^2\}.  
\end{equation}
This last set amounts to considering
vectors that are not ``too sparse'', \ie if $\bs u \in \amgis_{K_0}$ then
$\|\bs u\|_0 \geq K_0$, which determines our notation $\amgis_{K_0}$ as
opposed to $\Sigma_{K}$. However, the converse is not true and $\amgis_{K_0} \neq
\Sigma^\compl_{\lfloor K_0 \rfloor}$. Since belonging to
$\amgis_{K_0}$ prevents sparsity, we say that a vector $\bs u \in
\amgis_{K_0}$ is an \emph{anti-sparse} vector of level $K_0\geq 0$. 

\noindent Actually \eqref{eq:expect-qdist-on-antisparse} states that, for vectors $\bs x - \bs y \in
\amgis_{K_0}$, the expectation of $\Qdist(\bs x,\bs y)$ is close to
the one obtained with Gaussian random projections, \ie close to the
expectation $\gone\,\|\bs x - \bs y\|$ associated to $\kappasg =
0$. Thus, if we expect to show that, for all vectors $\bs x$ and
$\bs y$ in $\cl K$, $\Qdist(\bs x,\bs y)$
concentrates around $\gone\,\|\bs x - \bs y\|$, we must take into
account the anti-sparse nature of the difference $\bs x - \bs y$, \ie
we would need enforcing this vector to belong to $\amgis_{K_0}$ for a sufficiently large~$K_0$.
\medskip

Combining these three observations, and anticipating over the next
section, we can now refine the meaning
of~\eqref{eq:objective-paper}. We are actually going to show that, if $M$ is bigger than some $M_0$
growing with the typical dimension of $\cl K$ and decreasing with $\epsilon$
(see Sec.~\ref{sec:main-results}), then, with high probability,
$$ 
(\gone - \epsilon - \tfrac{\kappasg}{\sqrt K_0})\,\|\bs x
- \bs y\| - c\epsilon\delta\ \leq\ \textstyle \Qdist(\bs x,\bs y)\ \leq\ (\gone + \epsilon + \tfrac{\kappasg}{\sqrt K_0})\,\|\bs x
- \bs y\| + c\epsilon\delta,
$$
for all $\bs x, \bs y \in \cl K$ and $\bs x - \bs y \in \amgis_{K_0}$. 
\medskip

\noindent\emph{Remark:} As will be cleared later, our developments benefit of the tools and techniques
developed in \cite{plan2011dimension} where it is shown that, for a 1-bit
mapping $\bqmap': \bb R^N \to \{\pm 1\}^M$ such that $\bqmap'(\bs x)
= \sign(\bs \Phi \bs x)$ with a random Gaussian matrix $\bs \Phi \sim \cl N^{M\times N}(0,1)$, and for the normalized Hamming distance $\Qdist'(\bs x, \bs y) =
M^{-1} \sum_i \bbone[\qmap'_i(\bs x) \neq \qmap'_i(\bs y)]$, one has, provided $M \gtrsim \epsilon^{-4} w(\cl
K)^2$ and with probability exceeding $1-e^{-\epsilon^2 M}$, that for all $\bs x,\bs y \in \cl K$, 
$$
\textstyle |\Qdist'(\bs x,\bs y) - \arccos(\frac{\bs x^\transp\bs y}{\|\bs x\|\|\bs y\|})|\ \lesssim\ \epsilon. 
$$
Our extension to non-Gaussian sensing matrices is also inspired by
similar developments realized in \cite{ai2014one} for binary
mappings and other generalized linear models.

\section{Main Results}
\label{sec:main-results}

\subsection{Quasi-Isometric Quantized Embedding}
\label{sec:quasi-isometric-embed}

In regards to the context explained in the previous section, our first main result can be stated as follows.  
\begin{proposition}[Quantized sub-Gaussian quasi-isometric embedding]
\label{prop:main-result} Given $\delta>0$, $\epsilon \in (0,1)$, $K_0
> 0$, a bounded
subset $\cl K \subset \bb R^N$ and a sub-Gaussian distribution
$\normsg$ respecting
\eqref{eq:one-moment-sg-antisparse} for $0 \leq \kappasg < \infty$, there exist some
values $c,c'>0$, only depending on $\alpha$, such that, if 
\begin{equation}
  \label{eq:minimal-M-prop-embedding}
  M\ \gtrsim\ \tfrac{1}{\delta^2\epsilon^5} w(\cl K)^2,
\end{equation}
for a general set $\cl K$, or 
\begin{equation}
  \label{eq:minimal-M-prop-embedding-Ksparse}
M\ \gtrsim\ \tfrac{1}{\epsilon^2}\,\bar w(\cl K)^2\,
  \log(1+\tfrac{\|\cl K\|}{\delta \sqrt{\epsilon^3}}),
\end{equation}
for structured sets $\cl K$ (see
  Def.~\ref{def:structured-set} for the
  definition of $\bar w$), such as the set of bounded
  $K$-sparse signals or the one of bounded rank-$r$ matrices, then, for $\bs \Phi \sim \normsg^{M\times N}(0,1)$, a dithering $\bs \xi
\sim \cl U^{M}([0, \delta])$ and the associated quantized mapping
$\bs u \in \bb R^N \to \bqmap(\bs u) = \cl Q(\bs\Phi \bs u + \bs \xi)$, we have with probability at
least $1 - e^{-c'\epsilon^2M}$ and for all pairs $\bs x,\bs y \in \cl K$
with $\bs x - \bs y \in \amgis_{K_0}$, 
\begin{equation}
  \label{eq:quantized-quasi-isometry}
(\gone - \epsilon - \tfrac{\kappasg}{\sqrt K_0})\,\|\bs x
- \bs y\| - c\epsilon\delta\ \leq\ \textstyle \Qdist(\bs x,\bs y)\ \leq\ (\gone + \epsilon + \tfrac{\kappasg}{\sqrt K_0})\,\|\bs x
- \bs y\| + c\epsilon\delta.
\end{equation}
In the Gaussian case, \ie
for $\bs \Phi \sim \cl N^{M \times N}(0,1)$, the conditions remain the
same and
\eqref{eq:quantized-quasi-isometry} is simplified with $\kappasg = 0$, \ie there is no
  additional requirement on the anti-sparse nature of $\bs x - \bs y$
  in~\eqref{eq:quantized-quasi-isometry} since $K_0$ can be set to 1 and $\amgis_{K_0} = \bb R^N$.
\end{proposition}

In Prop.~\ref{prop:main-result}, as shown in Sec.~\ref{sec:framework}, the constant part ${\kappasg}/{\sqrt K_0}$ of the multiplicative
  distortion appearing in both sides of
  \eqref{eq:quantized-quasi-isometry} is unavoidable in the case of non-Gaussian projections (with
$\kappasg \neq 0$). Actually, we can show that this distortion cannot decay faster than 
 $\Omega(1/K_0)$ for non-Gaussian (but sub-Gaussian) random matrices when the level of anti-sparsity $K_0$ of $\bs x
- \bs y$ increases. To see this, it is
sufficient to study $\Qdist(\bs x,\bs y)$ for an asymptotically large
$M$, \ie $\bb E \Qdist(\bs x,\bs y)$ by the law of large numbers, and
to observe how the relative error between
$\bb E \Qdist(\bs x,\bs y)$ and $\gone \|\bs x -\bs y\|$ behaves when that
level $K_0$ increases.   

Taking $\delta = 1$ by simplicity, notice first that, from the observation made in \eqref{eq:abso-expec-dither-floor}, 
$$
\bb E \Qdist(\bs x,\bs y) = \bb E_{\bs \varphi} |\bs \varphi^\top (\bs x -
\bs y)| = \mu_{\rm sg}(\bs x - \bs y),
$$
where $\bs\varphi \sim
\cl B(\tinv{2})^N$ and $\mu_{\rm sg}$ was introduced in \eqref{eq:musg-def}. 

Let us then take $\bs x$ and
$\bs y$ such that the vector $\bs w := \bs x - \bs y$ is equal to 1 on its first $K_0$
components and zero elsewhere, \ie $\bs w \in \amgis_{K_0}$.  In this case and if $\bs \Phi$ is a
random Bernoulli matrix, $\mu_{\rm sg}(\bs w)$ is actually twice the \emph{mean absolute deviation} (MAD)
of a Binomial distribution ${\rm Bin}(K_0,\tinv{2})$ with $K_0$
degrees of freedom and success probability $p=1/2$ since 
$$
\textstyle \mu_{\rm sg}(\bs w) = \bb E |\!\sum_{j=1}^{K_0} \varphi_j| = 2\bb E
|(\sum_{j=1}^{K_0} X_j) - \tinv{2}K_0| = 2\bb E
|\beta_{K_0} - \bb E \beta_{K_0}|,  
$$
with, for $1\leq j\leq K_0$ and $X_j := \tinv{2} (\varphi_j + 1) \sim_{\rm iid} \cl B(\{0,1\},1/2)$ a Bernoulli random
variable such that $\bb P(X_j = 0) = 1/2$, and $\beta_{K_0} \sim {\rm
  Bin}(K_0,\tinv{2})$. 

However, from \cite{jveeh_stirling,gosper,blyth1980expected} we can
show that (see App.~\ref{app:bin-bound} for details) 
$$
\textstyle |\bb E |\beta_{K_0} - \bb E \beta_{K_0}| - \gone \tfrac{\sqrt K_0}{2}|
\geq C \tfrac{\sqrt K_0}{2} K_0^{-1},
$$
for $C = 1/7$. Consequently, for our choice of $\bs
w = \bs x - \bs y$ such that $\|\bs w\| = \sqrt K_0$, this shows that 
$$
\big| \bb E \Qdist(\bs x,\bs y) - \gone \|\bs x - \bs y\| \big|\ \geq\ 2C\, \|\bs x - \bs y\|\,K_0^{-1},
$$
and proves that, even if we reached an asymptotic regime in $M$, a
multiplicative distortion between $\Qdist(\bs x,\bs y)$ and $\gone \|\bs x - \bs y\|$
would remain, and this one could decay faster than $1/K_0$ when $K_0$
increases. It is therefore unclear if our decay in $1/\sqrt K_0$ is
optimal. 

To conclude this section, let us observe that
Prop.~\ref{prop:main-result} improves a proof of existence of
a quantized embedding given in \cite[Theorem
1.10]{plan2011dimension} where it was showed that, provided $M \gtrsim
\epsilon^{-12} w(\cl K - \cl K)^2$, there exists an arrangement of $M$
affine hyperplanes in $\bb R^N$ and a scaling factor $\lambda$ such
that
$$
|\lambda \Qdist_{\rm c}(\bs x,\bs y) - \|\bs x - \bs y\|| \leq
\epsilon, 
$$
where $\Qdist_{\rm c}$ denotes the fraction of affine hyperplanes
that separate the two vectors $\bs x$ and $\bs y$.

For reasons explained in Sec.~\ref{sec:proofs}, each element $\delta^{-1}|\qmap_i(\bs x) -
\qmap_i(\bs y)|_1 $ appearing in
$\delta^{-1}\Qdist(\bs x,\bs y) = \tinv{\delta M} \sum_{i=1}^M |\qmap_i(\bs x) -
\qmap_i(\bs y)|_1$ actually counts the number of parallel affine hyperplanes in $\bb
R^N$ normal to $\bs \varphi_i$ and far apart by $\delta$, with a
dithering that randomly displaces the origin. Therefore, 
Prop.~\ref{prop:main-result} basically constructs, in a random fashion, an arrangement of
$M$ such parallel hyperplane bundle, \ie in $M$ different
directions $\{\bs \varphi_i/\|\bs \varphi_i\|, i \in
[M]\}$. Considering a Gaussian matrix $\bs \Phi$ (with $\kappasg = 0$), we
have therefore proved that there with a minimal $M$ that grows like $\epsilon^{-5}$ rather
than $\epsilon^{-12}$ when $\epsilon$ decays (as expressed in
\eqref{eq:minimal-M-prop-embedding}). This is even reduced to
$\epsilon^{-2}$ for pairs of vectors taken in a structured set. 

\subsection{Consistency Width Decay}
\label{sec:consistency-width}

As a second important result, we optimize the decay law (as $M$
increases) of the distance of any pair of vectors $\bs x, \bs y \in \cl
K$ whose difference is ``not too sparse'' when those are 
mapped by $\bs A$ on the same
quantization point in $\delta\bb Z^M$, \ie when they are
\emph{consistent}. We refer to this distance as the \emph{consistency
  width} of $\bs A$.

This width could be characterized from Prop.~\ref{prop:main-result} when $\Qdist(\bs x,\bs y) =
0$, which provides $\|\bs x - \bs y\| \lesssim \epsilon \simeq
M^{-1/5}$ (or $M^{-1/2}$ if $\cl K$ is a structured set) for large $M$
respecting~\eqref{eq:minimal-M-prop-embedding} (resp. \eqref{eq:minimal-M-prop-embedding-Ksparse}), $\delta$ fixed and
$\kappasg/\sqrt K_0$ small. However, focusing on the conditions
guaranteeing the consistency of $\bs x$ and $\bs y$, and considering all
quantities fixed but $M$, our result below reaches the improved decay $\epsilon =
O(M^{-1/4})$ for a general set $\cl K$ and $\epsilon =
O(1/M)$ for a structured one. We prove the following proposition
in Sec.~\ref{sec:proof-consistency-width}.

\begin{proposition}[Consistency width upper bound]
\label{prop:consistency-width}
Let us take a quantization resolution $\delta >0$, an accuracy $\epsilon \in (0,1)$, a sub-Gaussian distribution $\normsg(0,1)$ respecting
\eqref{eq:one-moment-sg-antisparse} for $0 \leq \kappasg < \infty$,
$K_0>0$ such that $\sqrt K_0 \geq 16 \kappasg$ 
and a
bounded subset $\cl K \subset \bb B^{N}$ of $\bb R^N$.  For a value
$c>0$ depending only on $\alpha$, provided
\begin{equation}
  \label{eq:prop-consist-width-minimal-cond}
  M\ \gtrsim\ \tfrac{(2 + \delta)^4}{\delta^2 \epsilon^4} \, w(\cl K)^2    
\end{equation}
for a general set $\cl K$, or 
\begin{equation}
  \label{eq:prop-consist-width-minimal-cond-K-sparse}
  M\ \gtrsim\ \tfrac{2+\delta}{\epsilon}\, \bar w(\cl K)^2 \log\big(1+ 
  \tfrac{(2 + \delta)^{3/2} \|\cl K\|}{\delta \epsilon^{3/2}}\big),
\end{equation}
for a structured set $\cl K$, the map $\bqmap$ defined
  in \eqref{eq:psi-def} with $\bs \Phi \sim \normsg^{M\times N}(0,1)$
and $\bs \xi \sim
\cl U^M([0, \delta])$ is such that, with probability exceeding $1 - 2\exp(-c \epsilon M/(1+\delta))$,
\begin{equation}
  \label{eq:consistency-width}
  \bqmap(\bs x) = \bqmap(\bs y)\quad \Rightarrow\quad \|\bs x - \bs y\| \leq \epsilon, 
\end{equation}
for all $\bs x,\bs y \in
\cl K$ with $\bs x-\bs y \in \amgis_{K_0}$.
In the Gaussian case, \ie
for $\bs \Phi \sim \cl N^{M \times N}(0,1)$, the conditions above remain the
same with $\kappasg = 0$, \ie with no additional requirement on
the anti-sparse nature of $\bs x - \bs y$ in~\eqref{eq:consistency-width}.
\end{proposition}

Unfortunately, we were unable to produce a convincing counter example of
a pair of vectors both with 
difference not in $\amgis_{K_0}$ and failing to
meet \eqref{eq:consistency-width} under the conditions of Prop.~\ref{prop:consistency-width}. Therefore, it is not clear if the condition $\bs x - \bs y \in
\amgis_{K_0}$ is an artifact of the proof or if removing it could worsen then dependence in $\epsilon$ in
\eqref{eq:prop-consist-width-minimal-cond}.  

\section{Discussions and Perspectives}
\label{sec:discussions}

Before delving into the proofs of Prop.~\ref{prop:main-result} and
Prop.~\ref{prop:consistency-width} (see Sec.~\ref{sec:proofs} and
Sec.~\ref{sec:proof-consistency-width}, respectively), let us discuss
their meaning and limitations, providing also some perspectives for
future works.

\medskip
\noindent \textbf{On the impact of the diameter of structured
    sets:} For the structured sets considered in the Introduction, it
  is known that if the linear embedding \eqref{eq:JL-ineq} holds with
  high probability for all
  $\bs x,\bs y \in \cl K \subset \bb S^{N-1}$ with some distortion
  $\epsilon > 0$, then, since \eqref{eq:JL-ineq} is homogeneous, a
  simple rescaling argument proves that the same
  relation actually holds for all points in $\cl K' = \cup_{\lambda > 0} {\cl
    K}$, or equivalently for all points in the cone $\cl K'$ if $\cl K = \cl K'
  \cup \bb S^{N-1}$~\cite{baraniuk2008simple,mendelson2008uniform}. In particular, since such
  a linear embedding occurs with high probability for sub-Gaussian
  random matrices provided $M \gtrsim
  \epsilon^{-2}\,w(\cl K)^2$ \cite{mendelson2008uniform}, this
  requirement remains unchanged for reaching the embedding of vectors
  in $\cl K'$.  

Obviously, in the case of a quantized embedding such as
\eqref{eq:quantized-quasi-isometry}, the non-linear nature of $\cl Q$
prevents this rescaling argument from holding. However, an interesting phenomenon occurs anyway in this case
through the requirements \eqref{eq:minimal-M-prop-embedding-Ksparse}
and \eqref{eq:prop-consist-width-minimal-cond-K-sparse} of
Prop.~\ref{prop:main-result} and Prop.~\ref{prop:consistency-width},
respectively. Indeed, we see there that the diameter of the set $\cl
K$ has only a logarithmic impact on the minimal value of $M$ needed
for these propositions to hold, since $\bar w$
does not depend on the diameter of $\cl K$ (see
Def.~\ref{def:structured-set} and the subsequent explanations). 
This really slow increase approaches the scale-invariant
requirement obtained by linear embedding of structured sets, and is
anyway strikingly slower than the quadratic amplification of the
minimal number of measurements provided by \eqref{eq:minimal-M-prop-embedding} and 
\eqref{eq:prop-consist-width-minimal-cond} in the case of a general set $\cl K$,
as involved by
{(P\ref{p:w-hom})} when $\cl K$ is expanded like $\cl K \to \lambda
\cl K$ for $\lambda > 1$.    

\medskip
\noindent \textbf{Mitigating the anti-sparsity requirement:} For
both propositions, we can be concerned by the restriction that the vector difference must be ``not too sparse'', \ie for $\bs x,\bs y \in
\cl K$ there must be a sufficiently big $K_0$, either for having $\bs x -
\bs y \in \amgis_{K_0}$ and minimizing the distortion $\kappasg/\sqrt
K_0$ in \eqref{eq:quantized-quasi-isometry}, or for satisfying $\sqrt K_0 \geq 16 \kappasg$ in
{Prop.~\ref{prop:consistency-width}}. However, in certain cases, it is possible to adapt the sensing matrix as to increase this $K_0$.   

Indeed, assuming without loss of generality
that the vectors $\bs x -\bs y\in\cl K - \cl K$ are
expected to be ``too sparse'' only in $\bs \Psi =
\Id$ when the sensing matrix is
non-Gaussian (\ie $\kappasg \neq 0$), we can always
``rotate''\footnote{Strictly speaking, while $|{\rm det} \bs\Psi_0| =
  1$, $\bs \Psi_0 \in \cl O_N$ is a rotation only
if its determinant is $1$.} $\cl K$
with an ONB $\bs
\Psi_0$ of $\bb R^N$ so that elements of $\cl K' - \cl
K'$ with $\cl K' := \bs \Psi_0 \cl K$ have a higher anti-sparse degree than those of $\cl K - \cl K$, \ie
\begin{align}
&\max\{K_0: (\cl K' - \cl K') \cap \amgis_{K_0} \neq \emptyset\} = \min_{\bs u \in \cl K -
\cl K} \tfrac{\|\bs u\|^{2}}{\|\bs \Psi_0 \bs u\|^2_\infty}\nonumber\\
\label{eq:max-K0}
&\qquad\geq\quad \min_{\bs u \in \cl K -
\cl K} \tfrac{\|\bs u\|^{2}}{\|\bs u\|^2_\infty} = \max\{K_0: (\cl K - \cl K) \cap \amgis_{K_0} \neq \emptyset\}  
\end{align}
possibly trying to maximize the left hand side in the selection of
$\bs \Psi_0$. 

Therefore, while the requirements
imposed on $M$ in Prop.~\ref{prop:main-result} and Prop.~\ref{prop:consistency-width} are unchanged between $\cl K$
and $\cl K'$ in Prop.~\ref{prop:main-result} (by the 
invariance {(P\ref{p:w-invort})} of $w(\cl K)$ in Table~\ref{tab:Gaussian-mean-width-prop}) and since $\|\bs x' - \bs
y'\| = \|\bs x - \bs
y\|$ for $\bs x'=\bs\Psi_0 \bs x$ and $\bs y' = \bs \Psi_0 \bs y$,
``rotating'' $\cl K$
with $\bs \Psi_0$ helps to lighten the condition imposed on $\bs x-\bs y$.  
Moreover, this rotation is
of course equivalent to directly build a sensing matrix $\bs \Phi' = \bs \Phi
\bs \Psi_0$ to quasi-isometrically embed the set $\cl K$ with the
mapping $\bs A(\cdot) := \cl
Q(\bs \Phi' \cdot)$. Actually, in the case where $\bs
\Psi = \Id$ as above, a good choice for $\bs \Psi_0$ is the DCT
basis, \ie using the incoherence of those two bases that prevents a
sparse signal to be sparse in the frequency domain, also taking advantage of the fast FFT-based matrix-vector multiplication
offered by the DCT. Notice, however, that the procedure above cannot work if $\cl K$ is
expected to generate differences of vectors that are sparse in different
bases, \eg a union of incoherent bases such as $\Id$ and the DCT basis. In such
a case, it could be hard to maximize the right-hand side of~\eqref{eq:max-K0} over
$\bs \Psi_0$. 

Interestingly, a similar procedure to the one described
  above has been developed recently in \cite[Theorem 2.3]{oymak2016}
  in the context of fast circulant binary embeddings of finite sets of
  vectors. The requirement on the anti-sparse nature of the mapped
  vectors is there mitigated by taking $\bs \Psi_0$ as the product of
  a Hadamard transform with a diagonal matrix with random
  Rademacher entries, which can provably reduce the \emph{coherence}
  $\|\bs\Psi_0 \bs u\|^2_\infty/\|\bs u\|$ of too sparse $\bs u$ with high probability.

\medskip
\noindent \textbf{Intrinsic ``anti-sparse'' distortion limit:}  We can notice that for non-Gaussian random measurements, the
term $\kappasg/\sqrt{K_0}$ in \eqref{eq:quantized-quasi-isometry} is
actually lower bounded. This is simply due to the relation
$\|\bs u\|^2 \leq N \|\bs u\|^2_\infty$,
which implies $K_0 \leq N$ whatever the properties of the vector $\bs
u \in \cl K - \cl K \subset \bb R^N$. Consequently,  
$$
\tfrac{\kappasg}{\sqrt{K_0}} \geq \tfrac{\kappasg}{\sqrt{N}},
$$
which limits our hope to tighten the multiplicative error of quantized
non-Gaussian quasi-isometric embeddings, except if one considers asymptotic regimes
where $N$ can be considered as being much larger than $\kappasg^2$. 

\medskip
\noindent \textbf{Distortion regimes:}
As already noticed in \cite{jacques2013quantized},
Prop.~\ref{prop:main-result} allows us to distinguish different
regimes of the quasi-isometric embedding. If $\delta \simeq 0$, the
quantization operator tends to the identity function and
\eqref{eq:quantized-quasi-isometry} converges to a $\ell_2/\ell_1$
variant of the RIP generalized to any sets $\cl K$ and to
sub-Gaussian random matrices, as characterized in
  \cite{Schect06,plan2011dimension} for general sets and in \cite{Jacques2010} for
  sparse signal sets only. For $\delta \gg 2\|\cl K\|$ the
embedding becomes purely quasi-isometric and, keeping the context defined in Prop.~\ref{prop:main-result},
\eqref{eq:quantized-quasi-isometry} involves
\begin{equation}
  \label{eq:quantized-quasi-isometry-high-delta}
\gone \|\bs x
- \bs y\| - c(\epsilon\delta + \tfrac{\kappasg}{\sqrt K_0})\ \leq\ \textstyle \Qdist(\bs x,\bs y)\ \leq\ \gone \|\bs x
- \bs y\| + c(\epsilon\delta + \tfrac{\kappasg}{\sqrt K_0}),
\end{equation}
for some absolute constant $c>0$. However, in this case, the
quantization becomes essentially binary. In fact, it is exactly binary for
random matrices whose entries are generated from a bounded symmetric
sub-Gaussian distribution, \ie from $\varphi \sim \normsg(0,1)$ with
$\|\varphi\|_\infty \leq F$ for some $F > 0$. In this case, since $\cl K$
is assumed bounded, for all $\bs u \in \cl K$, $|(\bs \Phi \bs u)_i| \leq F \|\cl K\|$ and
the components of $\bs A(\bs u) = \cl Q(\bs \Phi \bs u + \bs \xi)$
with $\bs \xi \sim \cl U^M([0, \delta])$ can only take two values, \eg
$\{-1, 0\}$ if $0 \in \cl K$. Moreover, if $\varphi$ is unbounded and $0 \in \cl K$, its sub-Gaussian
nature is so that the fraction of quantized measurements that do not
belong to $\{-1,0\}$ can be made arbitrarily close to 0 when
$\delta$ increases.  In conclusion, similarly to
\cite{knudson2014one}, we have basically defined a one-bit quantized
embedding that preserves the norm of the projected vectors, as opposed
to the mapping $\bs A'(\cdot) = \sign(\bs \Phi \,\cdot)$ that loses
this information \cite{jacques2013robust,plan2013one}. Notice there
that the role of our dithering can be compared to the one of the
threshold inserted in the sign quantization
in~\cite{knudson2014one}. Conversely to that work, however, we do not provide any algorithm to reconstruct a signal
from its quantized mapping by $\bs A$.

\medskip
\noindent \textbf{Towards an $\ell_2/\ell_2$ quasi-isometric embedding?}
It is not clear if Prop.~\ref{prop:main-result} could be turned
into a quasi-isometric embedding between $(\cl K \subset \bb R^N, \ell_2)$ and $(\bs
A(\cl K) \subset \delta \bb Z^M, \ell_2)$. As said earlier, for
Gaussian random matrices and for finite sets $\cl K$, an approximate
quasi-isometric embedding can be found by integrating a non-linear distortion of
the $\ell_2$-distance, \ie in \eqref{eq:quantized-quasi-isometry} for
$\kappasg = 0$, $\|\bs x - \bs y\|$ is replaced by $g_\delta(\|\bs x - \bs
y\|)$ for some non-decreasing function $g_\delta:\bb R_+ \to \bb
R_+$. Interestingly, 
$|g_\delta(\lambda) - \lambda|=O(\sqrt{\delta\lambda})$ for
$\lambda \gg \delta$ and $|g_\delta(\lambda) - (\sqrt
  2\lambda/\sqrt\pi)^{1/2}|=O(\lambda)$ for $\lambda < \delta$, so that
  for small $\delta$ or large $\lambda$,  $g_\delta(\lambda) \approx
  \lambda$. Therefore, as soon as $\|\bs x - \bs y\| \gg \delta$, we
  get approximately a $\ell_2/\ell_2$ quasi-isometric
  embedding. Knowing if this extends to any subset $\cl K$ and to
  sub-Gaussian random matrices is left for a future work.

\medskip
\noindent \textbf{Reconstructing low-complexity vectors from quantized
    compressive observation?} Beyond the mere analysis of the quasi-isometric properties of
our quantized mapping and closer to the context of quantized
compressed sensing, this paper does not say anything on the
reconstruction algorithms that could be developed for recovering a
signal $\bs x$ from its observations $\bs z = \cl Q(\bs \Phi \bs
x)$. A few algorithms exist for realizing this operation, some when
$\delta$ is small compared to the expected dynamic of $\|\bs \Phi \bs
x\|$ \cite{Jacques2010,Dai2009,zymnis2010compressed}, others in the
1-bit CS setting
\cite{plan2013one,plan2012robust,bahmani2013robust,jacques2013robust}.
However, for the first category, their stability (or
convergence) does not rely on a quasi-isometric embedding property but
rather on the restricted isometry property
\cite{candes2006near,Dai2009,LasBouDav::2009::Demcracy-in-action} or
on variations involving other norms \cite{Jacques2013,Jacques2010}. In future research, it
will be appealing to find a proof of the instance optimality of those algorithms, \eg
for the basis pursuit dequantizer
(BPDQ), using the quasi-isometry property promoted by
Prop.~\ref{prop:main-result}, even if recent interesting results show that an optimal
``non-RIP'' proof can be developed for BPDQ \cite{Dirsksen-gap-RIP-sparse}.

\medskip
\noindent \textbf{Extension to fast and universal quantized embeddings?} We conclude this section by mentioning that it would be useful to prove
Prop.~\ref{prop:main-result} for structured random matrices, \eg for
random Fourier or random Hadamard ensembles~\cite{foucart2013mathematical},
as recently obtained in \cite{oymak2016} for the binary embedding of
finite sets. This would lead to
a fast computation of quantized mappings, with potential
application in nearest-neighbor search for databases of
high-dimensional signals. An open question is also the possibility to
extend this work to universally-quantized embeddings~\cite{BR_DCC13,SBV_SPIE13_Embeddings,B_TIT_12}, \ie
taking a periodic quantizer~$\cl Q$ in \eqref{eq:psi-def}. This could potentially lead to
quasi-isometric embeddings with (exponentially) decaying distortions on vectors sets with small
Gaussian width and using sub-Gaussian random matrices.

\section{On the necessity to dither the quantization}
\label{sec:nec-dither}

Considering the main results of this paper, namely
Prop.~\ref{prop:main-result} and Prop.~\ref{prop:consistency-width},
we could ask ourselves if a quantized mapping that would not include a
dithering could also verify \eqref{eq:quantized-quasi-isometry} and
\eqref{eq:consistency-width} under equivalent conditions on $M$ and on
the anti-sparse nature of $\bs x - \bs y$ for any vectors $\bs x,
\bs y$ in~$\cl K$. 

The answer is, however, negative in full generality, \ie it is possible
to define a quantized and undithered map $\bqmap: \bs x \to \cl Q(\bs \Phi \bs x)$ for
some appropriate quantizer resolution $\delta$ and sub-Gaussian random matrix $\bs \Phi$ that is \emph{incompatible} with
the definition of a quasi-isometric embedding with arbitrarily small additive
distortion or with an arbitrarily small consistency
width. 

To see this, let us set $\delta = 1$, $\cl Q(\lambda) :=
\argmin_{\lambda' \in \bb Z} |\lambda - \lambda'| = \lfloor \lambda +
\tinv{2} \rfloor$ (applied componentwise\footnote{It is easy, but slightly
  more technical, to adapt our development here to the quantizer $\cl
  Q(\cdot) = \delta \lfloor \cdot / \delta\rfloor$ defined in
  Sec.~\ref{sec:framework}. We thus prefer to select $\cl Q$ as a rounding
  operation for the sake of clarity.}), and take $\bs \Phi$ to be a
Bernoulli random matrix, \ie $\bs \Phi_{ij} \in \{\pm 1\}$. 
Given the value $\kappasg > 0$ associated to the distribution of $\bs
\Phi$, we also set arbitrarily an integer $K_0$ such that $\gone - (\kappa_{\rm
  sg}/{\sqrt K_0}) \geq 1/2$. In fact, we can compute 
that $\alpha = 1$ for a Bernoulli \rv, so that $\kappa_{\rm sg} \leq
9\sqrt{27} < 47$ from the bound given in
Sec.~\ref{sec:framework}. Therefore, $K_0 > (160)^2$ certainly works.   

We then define two $K_0$-sparse vectors $\bs u, \bs v \in \bb R^N$
with $\bs u$ equal to 1 on it first
$K_0$ components and 0 elsewhere, and $\bs v := (1+s K_0^{-1})\,\bs u$ for some fixed
$0 < |s|<1/2$. Clearly, when $K \geq K_0$ these two vectors belong to
the structured set $\cl K := \Sigma_{K}
\cap r_0 \bb B^N$ with $r_0 := \tfrac{3}{2} \sqrt K_0$. Moreover, from our definition of $K_0$, the difference vector $\bs w := \bs u - \bs
v = s K_0^{-1}\,\bs u$ is adjustably ``anti-sparse'' since it lies in
$\amgis_{K_0}$ with $\|\bs w\|_2^2/\|\bs w\|_\infty^2
= K_0$. Interestingly, $\bs u$ and $\bs v$ are also consistent with
respect to $\bqmap$ since $\cl Q(\bs \Phi \bs u) = \cl Q(\bs \Phi \bs v) = \cl Q(\bs \Phi \bs u +
s K_0^{-1} \bs \Phi \bs u)$.
This is due to the nature of quantization (\ie a rounding to the
closest integer) and to
the fact that both $\bs \Phi \bs u \in \bb Z^M$ and $\|s K_0^{-1} \bs \Phi \bs u\|_\infty \leq s < 1/2$.

Let us now assume, as involved by Prop.~\ref{prop:main-result}, that
for $\epsilon := \tfrac{s}{4(c + s) \sqrt K_0}$, it is possible to
find $M$ arbitrarily large before $\epsilon^{-2} \bar w(\cl K)^2 \log(1 +
\tfrac{r_0}{\epsilon})$ so that, with high probability and for all
$\bs x, \bs y \in \cl K$ with $\bs x - \bs y \in \amgis_{K_0}$, 
$$
(\gone - \epsilon - \tfrac{\kappasg}{\sqrt K_0})\,\|\bs x
- \bs y\| - c\epsilon\ \leq\ \textstyle \tinv{M} \|\cl Q(\bs \Phi \bs
x) - \cl Q(\bs \Phi \bs
y)\|_1,
$$  
with the constant $c>0$ defined in
\eqref{eq:quantized-quasi-isometry}.

However, by taking the consistent vectors $\bs x = \bs u$ and $\bs y =
\bs v$, this inequality leads by construction to   
$$
\textstyle 0 = \tinv{M} \|\bs A(\bs x) - \bs A(\bs y)\|_1\ \geq\ \big((\tfrac{2}{\pi})^{1/2} - \epsilon -
\tfrac{\kappa_{sg}}{\sqrt K_0}\big) \|\bs x - \bs y\| - c \epsilon
\geq (\tinv{2} - \epsilon) \|\bs x - \bs y\| - c \epsilon.
$$
In other words, since $\|\bs x - \bs y\| = s/\sqrt K_0 \leq s$ 
$$
\epsilon \geq \tfrac{1}{2} \tfrac{\|\bs x - \bs y\|}{c + \|\bs x - \bs
  y\|}  \geq \tfrac{s}{2(c + s) \sqrt K_0} = 2 \epsilon, 
$$
which is a clear contradiction.  We can similarly show that the same pair of consistent vectors $\bs x = \bs u$
and $\bs y = \bs v$ is incompatible with
Prop.~\ref{prop:consistency-width} as then the consistency width
cannot be arbitrarily small, even for asymptotically large $M$.   

\medskip
\noindent \emph{Remark:} Interestingly, the counter-example
above is easily hijacked to
show that it is impossible for the un-dithered quantized mapping $\bs
A(\cdot) := \cl
Q(\bs \Phi \,\cdot)$ to respect the following property for an arbitrarily small $\epsilon > 0$ and
provided $M$ is large enough,
$$
\big(C - \epsilon - g(K_0)\big)\,\|\bs x
- \bs y\| - c\epsilon\ \leq\ \textstyle h\big(\cl Q(\bs \Phi \bs
x), \cl Q(\bs \Phi \bs y)\big),\quad \forall \bs x, \bs y \in \cl K\ \text{with}\ \bs x - \bs y \in \amgis_{K_0},  
$$  
where $C,c>0$ are some universal constants, $h: \bb R^M \times \bb R^M \to
\bb R_+$ is any positive function vanishing on equal inputs (\eg a
norm, a pseudo-norm or any metric) and $g$ is any monotonically decreasing
function with $\lim_{t\to +\infty} g(t) = 0$. However, if $\cl Q$ is replaced by
a sign operator as in \cite{jacques2013robust,plan2011dimension}, then the known binary
$\epsilon$-stable embedding (or B$\epsilon$SE) relates the \emph{angular}
distance between $\bs x$ and $\bs y$ to the Hamming distance of their
mappings, \ie two distances that are equal to zero in our
counter-example above, which removes the contradiction.  

\medskip
\noindent \emph{Remark:} The question whether dithering is necessary in the special case of a
quantized mapping with a Gaussian random matrix $\bs \Phi$ remains
open.

\section{Proof of Proposition~\ref{prop:main-result}}
\label{sec:proofs}

The architecture of this proof is inspired by the one developed in
\cite{plan2011dimension} for characterizing a 1-bit random mapping $\bqmap': \bb R^N
\to \{\pm 1\}^M$, $\bs u\in \bb R^N \mapsto \bqmap'(\bs u) =
\sign(\bs\Phi\bs u)$. As will be clear below, some of the ingredients
developed there had of course to be adapted to the specificities of
$\bqmap$ and of our scalar quantization. Compared to
\cite{plan2011dimension} we have also paid attention to optimize
the dependency of $M$ to the desired level of distortions induced
by~$\bqmap$ in~\eqref{eq:psi-def}.  

Prop.~\ref{prop:main-result} is proved as a
special case of a more general proposition based on a ``softer''
variant of $\Qdist$. This new pseudo-distance is
established as follows. Defining the random mapping $\bs u \in \bb R^N
\mapsto \brmap(\bs u) := \bs \Phi \bs u + \bs \xi$, with $\rmap_i$ its $i^{\rm th}$ component,
we observe that for any $\bs
x,\bs y\in \bb R^N$, 
\begin{equation}
  \label{eq:Qdist-decomp-indicator}
  \textstyle \Qdist(\bs x,\bs y) = \tfrac{\delta}{M}\,\sum_{i=1}^M \sum_{k\in\bb Z} \bbone[{\cl
    E(\rmap_i(\bs x) - k\delta,\rmap_i(\bs y) - k\delta)}],  
\end{equation}
with the \emph{distinct sign event} $\cl E(a,b) := \{\sign a \neq
\sign b\}$. In words, for each $i\in [M]$, the sum over $k$ above
simply counts the number of thresholds in $\delta\Zbb$ separating $\rmap_i(\bs x) = \bs\varphi_i^\transp \bs x + \xi_i$ and $\rmap_i(\bs y) = \bs\varphi_i^\transp \bs y + \xi_i$ on the real line, since $\bbone[{\cl
  E(\rmap_i(\bs x) - k\delta,\rmap_i(\bs y) - k\delta)}]$ is equal to
1 for those and~0 for any other thresholds. 

Notice that the decomposition \eqref{eq:Qdist-decomp-indicator} also justifies
the observation made at the end of
Sec.~\ref{sec:consistency-width}, namely the existence of uniform random
tessellations of $\bb R^N$. Indeed, from the definition of $\bqmap$,
for each $i\in [M]$, $\sum_{k\in\bb Z} \bbone[{\cl
    E(\rmap_i(\bs x) - k\delta,\rmap_i(\bs y) - k\delta)}]$ also
  counts the number of parallel affine hyperplanes $\Pi_i := \{\bs u \in
  \bb R^N: \exists k \in \bb Z,\ \bs\varphi_i^\transp \bs u + \xi_i -
  k\delta = 0\}$, all normal to $\bs\varphi_i$ and $\delta/\|\bs\varphi_i\|$ far
  apart, separating $\bs
  x$ and $\bs y \in \bb R^N$.  In other words, $\bb R^N$ is here tessellated with
  multiple so-called ``hyperplane wave partitions'' $\{\Pi_i: i \in [M]\}$
  \cite{goyal_1998_lowerbound_qc,bib:Thao96} with random orientations,
  periods
  and dithered origin. 
  
\begin{figure}[t]
  \centering
  \includegraphics[width=.6\textwidth]{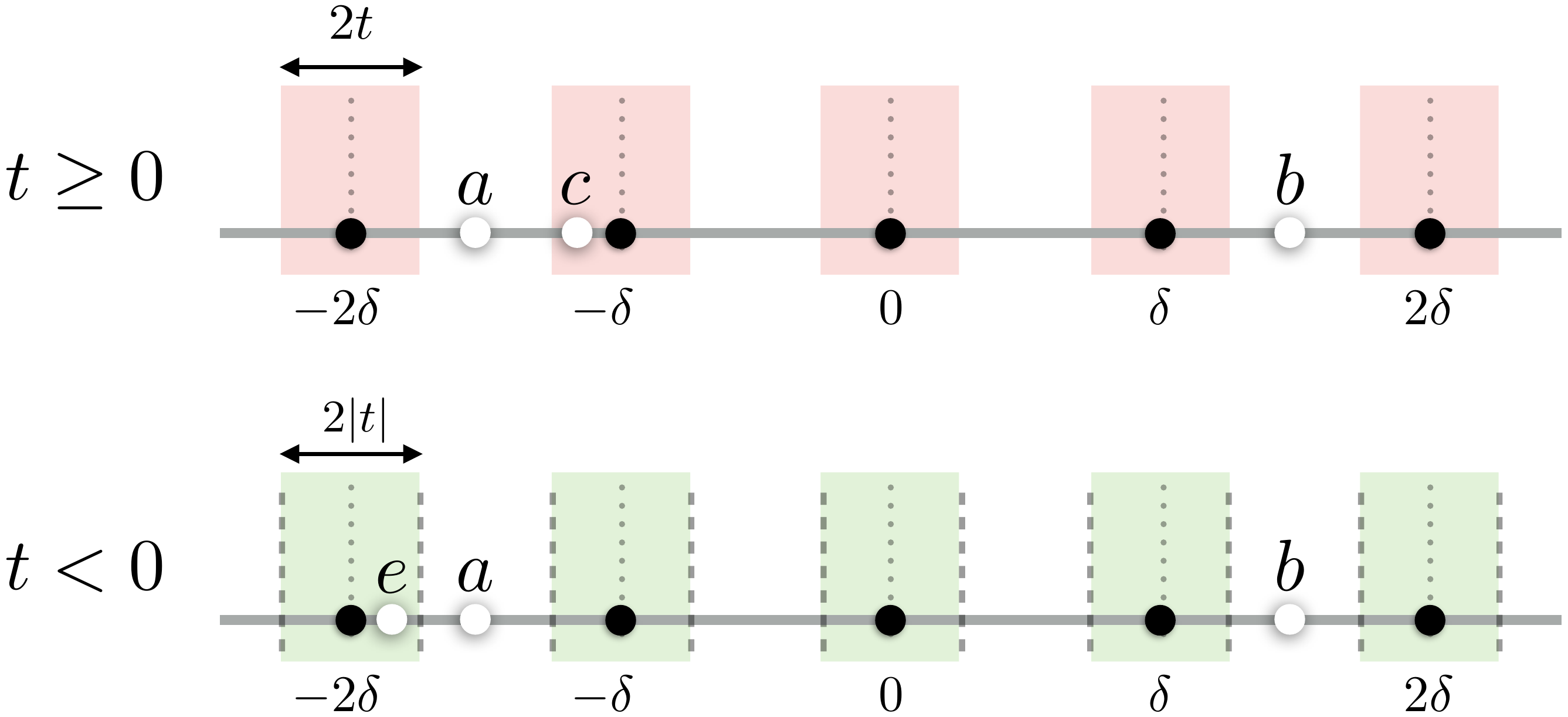}
  \caption{Behavior of the distance $d^t(a,b)$ for $a,b\in \bb R$. \emph{On the top}, $t\geq 0$ and
    forbidden areas determined by $\cl F^t$ are created when counting the number of thresholds
  $k\delta$ separating $a$ and $b$. For instance, for an additional
  point $c\in \bb R$ as on the figure, $d^0(a,b)=d^0(c,b)=3\delta$ but
  $3\delta=d^t(a,b)=d^t(c,b)+\delta \geq d^t(c,b)$ as $c$ lies in one
  forbidden area. \emph{On the bottom} figure, $t\leq 0$ and threshold
  counting procedure operated by $d^t$
  is relaxed. Now $d^t(a,b)$ counts the number of limits (in dashed) of the green
  areas determined by $\cl F^t$, recording only one per thresholds $k\delta$, that separate
  $a$ and $b$. Here, for $e\in \bb R$ as on the figure, $d^0(a,b)=d^0(e,b)=3\delta$ but $4\delta=d^t(e,b)=d^t(a,b) +\delta \geq d^t(a,b)$.}
  \label{fig:dt-behav}
\end{figure}

Based on this observation, and as a generalization of an
  equivalent distance given in 
\cite[Sec. 5]{plan2011dimension} for binary mappings, we introduce for some $t\in \bb R$ the new pseudo-distance 
\begin{equation}
  \label{eq:Dt-def}
  \textstyle \Qdist^t(\bs x,\bs y) := \tfrac{\delta}{M}\,\sum_{i=1}^M \sum_{k\in\bb Z} \bbone[{\cl
  F^t(\rmap_i(\bs x) - k\delta,\rmap_i(\bs y) - k\delta)}],  
\end{equation}
by defining the set
\begin{equation}
  \label{eq:def-Ft}
  \cl F^t(a,b) = \{a > t,\ b \leq -t\} \cup \{a < -t,\ b \geq t\}.  
\end{equation}
The pseudo-distance $\Qdist^t$ is a non-increasing function of $t$,
with $\cl F^0(a,b) = \cl E(a,b)$ and 
$$
\Qdist^{|t|}(\bs x,\bs y) \leq \Qdist(\bs x,\bs y) \leq \Qdist^{-|t|}(\bs x,\bs y).
$$
The behavior of
$\Qdist^t$  is
best understood by introducing the one-dimensional distance
\begin{equation}
  \label{eq:barD-def}
  \textstyle  d^t(a, b) := \delta\,\sum_{k\in\bb Z} \bbone[{\cl F^t(a
    - k\delta,b - k\delta)}]\ \in\ \delta \bb N,\quad {\rm for}\ a,b\in\bb R,  
\end{equation}
so that 
\begin{equation}
  \label{eq:link-QD-bard-def}
\textstyle\Qdist^t(\bs x,\bs y) = \tfrac{1}{M} \sum_{i=1}^M d^t(\rmap_i(\bs x),
\rmap_i(\bs y)).
\end{equation}
Fig.~\ref{fig:dt-behav} explains how $d^t(a,b)$ evolves for positive and
negative $t$, observing that, for each $k\in\bb Z$, ${\cl F^t(a
    - k\delta,b - k\delta)}$ determines forbidden or relaxed areas
  around the thresholds $k\delta$ separating $a$ and $b$ and counted by $d^t(a,b)$.
Moreover, the next Lemma, proved in
  App.~\ref{sec:proof-lemma-bounded-dbar}, provides a first evaluation
  of the impact of the distance ``softening'', by observing that, essentially, $d^t(a,b)$ is
  not
very far from both $|a - b|$ and $d^s(a,b)$ for $s$ close to $t$.
\begin{lemma}
\label{lem:bounded-dbar}
For any $a,b\in \bb R$ and $t,s\in \bb R$,\quad 
\begin{align}
\label{eq:diff-dt-ds}
\big|d^t(a, b) - d^s(a, b)\big|&\leq\ 4(\delta + |t-s|),\\
\label{eq:diff-dt-dist}
\big|d^t(a, b) - |a - b|\big|&\leq\ 4(\delta + |t|).    
\end{align}
\end{lemma}

As announced above, we aim now at proving the next proposition whose special case $t=0$ leads to Prop.~\ref{prop:main-result}.
\begin{proposition}
\label{prop:main-relaxed-result}
Given $\delta>0$, $\epsilon \in (0,1)$, $t\in \bb R$, $K_0 > 0$, a
bounded subset $\cl K \subset \bb R^N$ and a sub-Gaussian distribution
$\normsg$ respecting
\eqref{eq:one-moment-sg-antisparse} for $0 \leq \kappasg < \infty$, there exist some
values $C,c,c'>0$, only depending on $\alpha$, such that, if
\begin{equation}
  \label{eq:cond-M-soften-embed}
  M\ \geq C \max(\epsilon^{-2} \cl H(\cl K, \sqrt{\delta^2\epsilon^3}), \tinv{\delta^2\epsilon^{3}}\,w(\cl
  K_{\sqrt{\delta^2\epsilon^3}})^2), 
\end{equation}
with $\cl H(\cl
K,\eta)$ the Kolmogorov $\eta$-entropy of $\cl K$ and the local set $\cl
K_\eta := (\cl K - \cl K) \cap \eta\,\bb B^N$ for $\eta>0$, then
for $\bs \Phi \sim \normsg^{M\times N}(0,1)$, a dithering $\bs \xi
\sim \cl U^{M}([0, \delta])$, and the associated mapping $\bqmap$
defined in \eqref{eq:psi-def},  we have with probability exceeding $1 - e^{-c\epsilon^2M}$ that for all pairs $\bs x,\bs y \in \cl K$
with $\bs x - \bs y \in \amgis_{K_0}$,
\begin{equation}
  \label{eq:quantized-quasi-isometry-with-t}
\big|\ \Qdist^t(\bs x,\bs y) - \gone \|\bs x - \bs
y\|\ \big|\ \leq\ 
(\epsilon + \tfrac{\kappasg}{\sqrt{K_0}})\|\bs x
- \bs y\|\ +\ c' (|t| + \delta\epsilon).
\end{equation}
\end{proposition}

\begin{proof}
The proof sketch of
Prop.~\ref{prop:main-relaxed-result} is as follows: \emph{(i)} given
$\bs x, \bs y \in \bb R^N$, we first show that
the \rv $\Qdist^t(\bs x,\bs y)$ concentrates with high probability around $\gone \|\bs x - \bs
y\|$ up to a systematic bias $\tfrac{\kappasg}{\sqrt K_0})\|\bs x -
\bs y\|$ due to the sub-Gaussian nature of $\bs \Phi$ and controlled
by the anti-sparse level of $\bs x -\bs y$; \emph{(ii)} we take a finite covering of $\cl K$ by a $\eta$-net $\cl
G_\eta \subset \cl K$ (for $\eta>0$) and we
extend the concentration of $\Qdist^t(\bs x,\bs y)$ to all vectors of $\cl
G_\eta$ by union bound; \emph{(iii)} we show that the softened
pseudo-distance $\Qdist^t$ is
sufficiently continuous in a neighborood of each pair of vectors in $\cl
G_\eta$, which then allows us to extend \eqref{eq:quantized-quasi-isometry-with-t}
to all pair of vectors in $\cl K$, as stated by Prop.~\ref{prop:main-relaxed-result}.   

\paragraph*{\em (i) Concentration of $\Qdist^t(\bs x,\bs y)$:} 
Given a fixed pair $\bs x,\bs y \in \bb R^N$, we show that
$\Qdist^t(\bs x,\bs y)$ concentrates around its mean by bounding  
its sub-Gaussian norm as defined in \eqref{eq:sub-gaussian-def}. 
From \eqref{eq:link-QD-bard-def}, $\Qdist^t(\bs x,\bs y) = M^{-1} \sum_i Z_i^t$
with the $M$ random variables
$Z_i^t := d^t(\bs\varphi_i^\transp\bs x + \xi_i,\bs\varphi_i^\transp\bs y + \xi_i)$
for $1\leq i\leq M$. However, the sum of $D$ independent sub-Gaussian random variables $\{X_1,
\,\cdots, X_D\}$ is approximately invariant under rotation \cite{IntroNonAsRandom},
which means that
\begin{equation}
  \label{eq:approx-rot-inv}
  \|\sum_i (X_i - \bb E X_i)\|^2_{\psi_2}\ \lesssim\ \sum_i \|X_i - \bb E X_i\|^2_{\psi_2}.  
\end{equation}
Therefore, from~\eqref{eq:approx-rot-inv}, we find
\begin{equation}
  \label{eq:sum-zi-subgauss-norm}
\textstyle\big\|\sum_{i=1}^M  (Z_i^t - \bb E Z_i^t)\big\|^2_{\psi_2}\ \lesssim\
\sum_{i=1}^M  \|Z_i^t - \bb E Z_i^t\|^2_{\psi_2} = M \|Z_1^t - \bb
E Z_1^t\|^2_{\psi_2}\ \leq 4 M \|Z_1^t\|^2_{\psi_2}.  
\end{equation}

As shown in the following lemma (proved in
App.~\ref{sec:proof-lemma-subGaussian-Dt} by using Lemma~\ref{lem:bounded-dbar}) 
$\|Z_1^t\|_{\psi_2}$ can be upper bounded (and with it, the sub-Gaussian norm of $\Qdist^t(\bs x,\bs y)$).
\begin{lemma}
\label{lem:subGaussian-Dt}
Let us take $\bs \varphi \sim \normsg^N(0,1)$ and $\xi \sim \cl U([0,\delta])$. For
a fixed $t\in \bb R$, the random variable $Z^t := d^t(\bs\varphi^\transp\bs x + \xi,\bs\varphi^\transp\bs y + \xi)$ is
sub-Gaussian with $\psi_2$-norm bounded by 
\begin{equation}
  \label{eq:psi2-norm-Zt}
\|Z^t\|_{\psi_2} \lesssim \delta + |t| + \|\bs x
- \bs y\|.  
\end{equation}
Moreover,
\begin{equation}
  \label{eq:expect-diff-Zt-Z0}
\big|\bb E\,Z^t - \musg(\bs x - \bs y)|\big| \
\lesssim\ |t|,  
\end{equation}
with $\musg(\bs x - \bs y) = \bb E |\scp{\bs\varphi}{\bs x- \bs y}| = \gone \|\bs x - \bs y\|$ if $\bs \varphi \sim \normpdf^N(0,1)$.
\end{lemma}

Consequently, from \eqref{eq:sum-zi-subgauss-norm} and \eqref{eq:psi2-norm-Zt}, $X := \tinv{\sqrt{M}}\sum_{i=1}^M  (Z_i^t - \bb E Z_i^t)$ is itself
sub-Gaussian with $\|X\|_{\psi_2} \lesssim \delta + |t| + \|\bs x
- \bs y\|$.
Therefore, from the tail bound \eqref{eq:sg-tail-bound}, there exists a
$c>0$ such that for any~$\epsilon>0$
$$
\bb P\big[\big|\tinv{M} \sum_i (Z_i^t - \bb E Z_i^t)\big| > \epsilon\,(\delta + |t| + \|\bs x
- \bs y\|)\big]
\leq 2 \exp\big(- c \epsilon^2 M\big).
$$
Since $\bb
E Z_i^t=\bb
E Z_1^t$ and
$\bb E Z_i^0 = \bb E|\scp{\bs \varphi}{\bs x-\bs y}| = \musg(\bs x - \bs y)$ for all $i\in
[M]$, \eqref{eq:expect-diff-Zt-Z0} provides
\begin{align*}
\textstyle|\tinv{M} \sum_i (Z^t_i - \bb E Z^t_i)\big|
&\textstyle= |\Qdist^t(\bs x,\bs y) - \bb E Z^t_1 |\\
&\textstyle\geq |\Qdist^t(\bs x,\bs y) - \musg(\bs x - \bs y) | - |\bb E Z^t_1 - \musg(\bs x - \bs y)|\\
&\textstyle\geq |\Qdist^t(\bs x,\bs y) - \musg(\bs x - \bs y) | - c'|t|,  
\end{align*}
for some constant $c'>0$, and
\begin{equation}
  \label{eq:concent-soften-dt-on-fixe-pair}
\textstyle\bb P\big[\big| \Qdist^t(\bs x,\bs y) - \musg(\bs x - \bs y) \big| >
  c'|t| + \epsilon\,(\delta + |t| + \|\bs x
  - \bs y\|)\big]
  \leq 2 \exp\big(- c \epsilon^2 M\big).  
\end{equation}

\paragraph*{\em  (ii) Extension to a covering of $\cl K$:}  
Given a radius $\eta>0$ to be specified later, let $\cl G_\eta$ an $\eta$-net of $\cl K$, \ie a finite vector set such that for any $\bs x \in \cl
K$ there exists a $\bs x_0 \in \cl G_\eta$ with $\|\bs x - \bs x_0\| \leq
\eta$. In particular, any vectors $\bs x,\bs y \in \cl K$
can then be written as
\begin{equation}
  \label{eq:x-y-approx-in-net}
  \bs x = \bs x_0 + \bs x',\quad \bs y = \bs y_0 + \bs y',
\end{equation}
for some $\bs x_0,\bs y_0 \in \cl G_\eta$ and $\bs x',\bs y' \in (\cl K -
\cl K)\cap \eta \bb B^{N}$. We also assume that the size of $\cl G_\eta$ is
minimal so that, by definition, $\log |\cl G_\eta| = \cl H(\cl
K,\eta)$, with $\cl H$ the \emph{Kolmogorov
$\eta$-entropy} of $\cl K$.

Since there are no more than $|\cl G_\eta|^2$ distinct pairs of
vectors in $\cl G_\eta$, given $t\in \bb R$, a standard union bound over
\eqref{eq:concent-soften-dt-on-fixe-pair} shows that there exist some constant
$C,c',c''>0$ such that, if $M \geq C \epsilon^{-2} \cl H(\cl K,\eta)$ 
\begin{align}
&\textstyle\bb P\big[\ \forall \bs x_0, \bs y_0 \in \cl G_\eta,\ \big|  \Qdist^t (\bs x_0,\bs y_0) - \musg(\bs x_0 - \bs y_0) \big| \leq
  c'|t| + \epsilon\,(\delta + |t| + \|\bs x_0
- \bs y_0\|)\big]\nonumber\\
\label{eq:concent-Dt-cover-set}
&\textstyle\geq\ 1\, -\, |\cl G_\eta|^2 \exp\big(- c\epsilon^2 M\big)\ \geq\ 1 - 2 \exp\big(- c''\epsilon^2 M\big).
\end{align}

\paragraph*{\em  (iii) Extension to $\bb R^N$ by continuity of $\Qdist^t$:}

We can extend the event characterized in
\eqref{eq:concent-Dt-cover-set} to all pairs of vectors
in $\cl K$ by analyzing the continuity property of $\Qdist^t$ in
a limited neighborhood around the considered vectors. We propose here
to analyze this continuity with respect to $\ell_2$-perturbations of
those vectors, as
compared to $\ell_1$-perturbations in \cite{plan2011dimension}. As
will be clearer later, this allows us to reach a better control over $M$ with respect to $\epsilon$.

\begin{lemma}[Continuity with respect to $\ell_2$-perturbations] 
\label{lem:continuity-L2-Dt} 
Let $\bs x_0,\bs y_0,\bs x',\bs y' \in \bb R^N$. We assume
  that $\|\bs \Phi \bs x'\| \leq \eta \sqrt M$, $\|\bs \Phi \bs y'\|
  \leq \eta \sqrt M$ for some $\eta >0$. Then for every $t\in \bb R$ and
  $P\geq 1$ one has
  \begin{equation}
    \label{eq:continuity-L1-Dt}
    \textstyle \Qdist^{t+\eta\sqrt P}(\bs x_0,\bs y_0) - 4(\tfrac{\delta}{P} + \tfrac{\eta}{\sqrt P})\ \leq\ 
    \Qdist^t(\bs x_0 + \bs x',\bs y_0 + \bs y')\ \leq\ 
    \Qdist^{t-\eta\sqrt P}(\bs x_0,\bs y_0) + 4(\tfrac{\delta}{P} + \tfrac{\eta}{\sqrt P}).
  \end{equation}  
\end{lemma}
The proof is given in App.~\ref{sec:app-proof-cont-L2-perturb}.
Interestingly, the following proposition proved in App.~\ref{sec:proof-prop-small-resid-l2} shows that $\|\bs \Phi \bs x'\|$ and $\|\bs \Phi
\bs y'\|$ can indeed be bounded uniformly for all $\bs x',\bs y' \in \cl
K_\eta := (\cl K -
\cl K)\cap \eta \bb B^N$. 
\begin{lemma}[Diameter stability under random projections]
\label{lem:small-resid-l2}
Let $\cl R \subset \bb R^N$ be bounded, \ie $\|\cl R\|:=\sup_{\bs u
  \in \cl R} \|\bs u\| < \infty$ and assume
$\cl R \ni \bs 0$. Then, for some $c>0$, if 
$$
M \gtrsim \tfrac{\alpha^4 w(\cl R)^2}{\|\cl R\|^{2}},
$$
for $\bs \Phi \sim \normsg^{M \times N}(0,1)$ and with probability at least
$1-\exp(-c\,\alpha^{-4}M)$, we have for all $\bs x\in\cl R$
\begin{equation}
  \label{eq:small-resid-l2}
  \tinv{M} \|\bs \Phi\bs x\|^2 \leq \|\cl R\|^2,  
\end{equation}
\ie $\|\bs\Phi \cl R\| \leq \sqrt M \, \|\cl R\|$.
\end{lemma}

For the sake of simplicity, we consider below the sub-Gaussian parameter
$\alpha$ as fixed and integrate it in explicit or
hidden constants, as in the notations ``\,$\lesssim$\,'' or ``\,$\gtrsim$\,''. Noting that $\|\cl K_\eta\| \leq \eta$ and using a union
bound over \eqref{eq:concent-Dt-cover-set} and \eqref{eq:small-resid-l2}, we get that if 
$$
M \gtrsim
\max(\epsilon^{-2} \cl H(\cl K, \eta), \eta^{-2} w(\cl
K_\eta)^2),
$$
with probability higher
than $1- 4\exp(-c'\epsilon^2 M)$, for all $\bs x_0,\bs y_0 \in \cl
G_\eta$ and all $\bs x',\bs y' \in \cl K_\eta$,
\begin{eqnarray}
  &\label{eq:t-Peta-bound}
  \textstyle \big|\Qdist^{t-\eta \sqrt P}(\bs x_0,\bs y_0) - \musg(\bs x_0 - \bs y_0) \big| \leq c|t-\eta \sqrt P| + \epsilon\,(\delta +
    |t-\eta \sqrt P| + \|\bs x_0
- \bs y_0\|),\\
  &\label{eq:tpPeta-bound}
  \textstyle \big|\Qdist^{t+\eta \sqrt P}(\bs x_0,\bs y_0) - \musg(\bs x_0 - \bs y_0) \big| \leq c|t+\eta \sqrt P| + \epsilon\,(\delta + |t+\eta \sqrt P| + \|\bs x_0
- \bs y_0\|),\\
&\label{eq:L2-xp-yp-bound}
\|\bs \Phi \bs x'\|_2 \leq
\eta \sqrt M,\quad \|\bs \Phi \bs y'\|_2 \leq \eta \sqrt M,
\end{eqnarray}
for some $C,c,c'>0$ depending only on $\alpha$.

Therefore, for any $\bs x, \bs y \in \cl K$, using sequentially
\eqref{eq:x-y-approx-in-net}, \eqref{eq:L2-xp-yp-bound},
the upper bound given in Lemma~\ref{lem:continuity-L2-Dt} and \eqref{eq:t-Peta-bound} provides
\begin{gather*}
\textstyle \Qdist^t(\bs x, \bs y)\ \leq\ \Qdist^{t-\eta \sqrt P}(\bs x_0, \bs y_0) +
4(\tfrac{\delta }{P} + \tfrac{\eta}{\sqrt P})\\
\leq (c+\epsilon) |t-\eta \sqrt P| + \musg(\bs x_0 - \bs y_0)  +\epsilon \|\bs x_0
- \bs y_0\| + \epsilon\delta + 4(\tfrac{\delta }{P} + \tfrac{\eta}{\sqrt P}).
\end{gather*}
However, given $\bs \varphi \sim \normsg^N(0,1)$, using Jensen's
inequality, the reverse triangular inequality and~\eqref{eq:bound-musg}, we find
\begin{align*}
\big|\musg(\bs x_0 - \bs y_0) - \musg(\bs x - \bs y) \big| 
\ &=\ 
\big|\bb E |\scp{\bs \varphi}{\bs x_0 - \bs y_0}| - \bb E |\scp{\bs \varphi}{\bs x
  - \bs y}|\big|\\
&\leq\ \bb E |\scp{\bs \varphi}{\bs x'}| + \bb E |\scp{\bs \varphi}{\bs
  y'}|\ \leq\ 2\eta.  
\end{align*}
Moreover, $|\|\bs x_0
- \bs y_0\| - \|\bs x
- \bs y\|| \leq 2\eta$, so that,
\begin{align*}
\textstyle \Qdist^t(\bs x, \bs y) - \musg(\bs x - \bs y)&\textstyle \leq
\epsilon \|\bs x
- \bs y\| + (c+\epsilon)(|t| + \eta \sqrt P) + 2\eta + 2\epsilon\eta
+ \epsilon\delta + 4(\tfrac{\delta }{P} + \tfrac{\eta}{\sqrt P}).
\end{align*}
If $\bs x - \bs y \in \amgis_{K_0}$, then 
\eqref{eq:expect-qdist-on-antisparse} induces 
${|\musg(\bs x - \bs y) - \gone \|\bs x - \bs y\|| \leq
\kappasg\|\bs x - \bs y\|/\sqrt{K_0}}$
and assuming~$\epsilon < 1$, there exists a $c>0$ such that
\begin{align}
\label{eq:bidon}
\textstyle \Qdist^t(\bs x, \bs y) - \gone \|\bs x - \bs y\|&\textstyle \leq
(\epsilon + \tfrac{\kappasg}{\sqrt{K_0}})\|\bs x
- \bs y\| + c ( |t| + \eta \sqrt P + \eta
+ \epsilon\delta + \tfrac{\delta }{P} + \tfrac{\eta}{\sqrt P} ).
\end{align}

Taking $P=\epsilon^{-1} \geq 1$ and $\eta=\delta\epsilon^{3/2} <
\delta \epsilon$, which gives $\eta \sqrt P = \delta \epsilon$ and
$\eta/\sqrt P = \delta \epsilon^2 \leq \delta \epsilon$, we find for
another $c>0$
\begin{align*}
\textstyle \Qdist^t(\bs x, \bs y) - \gone \|\bs x - \bs y\|&\textstyle \leq
(\epsilon + \tfrac{\kappasg}{\sqrt{K_0}})\|\bs x
- \bs y\| + c (|t| + \delta\epsilon).
\end{align*}

Similarly, using
\eqref{eq:x-y-approx-in-net}, \eqref{eq:L2-xp-yp-bound},
the lower bound given in Lemma~\ref{lem:continuity-L2-Dt} and \eqref{eq:tpPeta-bound}, we obtain 
\begin{align*}
\textstyle \Qdist^t(\bs x, \bs y) - \gone \|\bs x - \bs y\| &\textstyle \geq\
- (\epsilon + \tfrac{\kappasg}{\sqrt{K_0}})\|\bs x
- \bs y\|\ -\ c (|t| + \delta\epsilon).
\end{align*}
Finally, we have thus shown that there exist some $c,c'>0$ such that for 
\begin{equation}
  \label{eq:cond-on-M}
  M\ \gtrsim\
  \max(\epsilon^{-2} \cl H(\cl K, \sqrt{\delta^2\epsilon^3}), \tinv{\delta^2\epsilon^{3}}\,w(\cl
  K_{\sqrt{\delta^2\epsilon^3}})^2),  
\end{equation}
with probability
at least $1-4 \exp(-c' \epsilon^2 M)$ the bound
\begin{align*}
\textstyle \big| \Qdist^t(\bs x, \bs y) - \gone \|\bs x
- \bs y\| \big|&\textstyle\ \leq\ (\epsilon + \tfrac{\kappasg}{\sqrt{K_0}})\|\bs x
- \bs y\|\ +\ c (|t| + \delta\epsilon)
\end{align*}
holds for all $\bs x,\bs y \in \cl K \cap \amgis_{K_0}$, which
finishes the proof of Prop.~\ref{prop:main-relaxed-result}.
\end{proof}

As mentioned earlier, Prop.~\ref{prop:main-result} is thus obtained by
simplifying the requirement \eqref{eq:cond-M-soften-embed} appearing in
Prop.~\ref{prop:main-relaxed-result}. First, for a general bounded set $\cl K$,
since the Sudakov inequality in (P\ref{p:w-suda}) provides $\cl H(\cl K,
\eta) \lesssim \tfrac{w(\cl K)^2}{\eta^2}$, noticing that $\cl K_\eta \subset (\cl K - \cl
K)$ and that (P\ref{p:w-setinc}) and (P\ref{p:w-setdiff}) provide 
$w(\cl K_\eta) \leq w(\cl K - \cl K) \leq 2 w(\cl K)$, we deduce
that \eqref{eq:cond-on-M}
holds if   
$$
M \gtrsim \tfrac{1}{\delta^2\epsilon^5} w(\cl
  K)^2,
$$
as imposed in \eqref{eq:minimal-M-prop-embedding}.

Second, in the case of a quantized embedding of the structured sets defined in the Introduction (see
Def.~\ref{def:structured-set}), we can even reach a much weaker
condition on $M$. Indeed, for such a set $\cl K$ with $d=\|\cl K\|$,
from~\eqref{eq:structured-set-prop-local} and the definition of $\bar w$, we
have for any $\eta > 0$
$$
w(\cl K_{\eta})^2 = w\big( (\cl K - \cl K) \cap \eta \bb B^N\big)^2 = d^2 w\big(
(d^{-1}\cl K - d^{-1}\cl K) \cap (d^{-1}\eta \bb B^N) \big)^2 \leq \eta^2\,
\bar w(\cl K)^2,
$$
so that, from~\eqref{eq:structured-set-prop-kolmo}, the
right-hand side of \eqref{eq:cond-M-soften-embed} can be bounded as
\begin{align*}
\textstyle \max(\epsilon^{-2} \cl H(\cl K, \sqrt{\delta^2\epsilon^3}), \tinv{\delta^2\epsilon^{3}}\,w(\cl
K_{\sqrt{\delta^2\epsilon^3}})^2)&\textstyle \leq   
\max(\epsilon^{-2} \bar w(\cl K)^2 \log(1+\frac{\|\cl K\|}{\sqrt{\delta^2\epsilon^3}}), \bar w(\cl K)^2)\\
&\textstyle  \leq \epsilon^{-2} \bar w(\cl K)^2 \log(1+\frac{\|\cl K\|}{\sqrt{\delta^2\epsilon^3}}).   
\end{align*}
This explains the simpler
requirement~\eqref{eq:minimal-M-prop-embedding-Ksparse} needed
  for structured sets in Prop.~\ref{prop:main-result}. 
\medskip

\noindent \emph{Example:} Let us conclude this section by deducing  an
upper bound on $\bar w^2(\cl K)$ for the set $\cl K := \Sigma^{\bs \Psi}_K\cap d\,\bb
B^N$ (with $d  = \|\cl K\| > 0$) of bounded
$K$-sparse vectors in an orthonormal basis $\bs \Psi \in \bb R^{N \times N}$ of $\bb R^N$. 
We first notice that since ${w({\Sigma^{\bs \Psi}_K \cap
d\,\bb B^N}) = w(\Sigma_K \cap
d\,\bb B^N)}$ by invariance over the orthogonal group $\cl O_N$ (see {(P\ref{p:w-invort})}
in Table~\ref{tab:Gaussian-mean-width-prop}) and from {(P\ref{p:w-spars})},
$$
w(\cl K/\|\cl K\|)^2 \lesssim K \log N/K.
$$ 
Moreover, the Kolmogorov entropy is also invariant under
$\cl O_N$, \ie $\cl H(\Sigma_K \cap
d\,\bb B^N,\eta) = \cl H(\Sigma^{\bs \Psi}_K \cap
d\,\bb B^N,\eta)$ and it is known
that (see, \eg
\cite{DLBB_ECETR09})
$$
\textstyle \cl H(\Sigma_K \cap
d\,\bb B^N, \eta)
\lesssim \log({N \choose K}(1 + \tfrac{2d}{\eta})^K) \leq
K \log(\tfrac{e N}{K}(1 + \tfrac{2d}{\eta})) 
\ \lesssim
K \log(\tfrac{N}{K}) \log(1+\tfrac{d}{\eta}),
$$ 
by using Stirling's bound. This shows that $\cl H(\cl K, \eta) \leq \bar
w({\cl K})^2 \log(1+ \frac{d}{\eta})$ with $\bar
w({\cl K})^2 \lesssim K
\log N/K$. Additionally,
since $\Sigma^{\bs \Psi}_K$ is invariant under dilation, $d^{-1}\cl K
- d^{-1}\cl K \subset \Sigma^{\bs \Psi}_K - \Sigma^{\bs \Psi}_K
\subset \Sigma^{\bs \Psi}_{2K}$ and 
\begin{align*}
w\big((d^{-1}\cl K
- d^{-1}\cl K)\cap \epsilon \bb B^N)^2&\leq w\big(\Sigma^{\bs
  \Psi}_{2K} \cap \epsilon \bb B^N)^2 =\epsilon^2 w\big(\Sigma^{\bs
  \Psi}_{2K} \cap \bb B^N)^2\\
&\lesssim \epsilon^{2} 2K \log
(N/2K) \lesssim \epsilon^{2} K \log
(N/K),
\end{align*}
showing again, by matching with~\eqref{eq:structured-set-prop-local}, that we have $\bar w(\cl K)^2 \lesssim K \log
(N/K)$. 

This confirms that $\bar w(\cl K)^2$ has the same upper bound than $w(\cl
K/\|\cl K\|)^2$. Therefore, for the structured set $\cl K$ of bounded $K$-sparse vectors,
\eqref{eq:cond-on-M} (and therefore \eqref{eq:minimal-M-prop-embedding-Ksparse}) is then satisfied~if
$$
M \gtrsim \tfrac{1}{\epsilon^2} K \log(\tfrac{N}{K})
\log(1+\tfrac{\|\cl K\|}{\delta \sqrt{\epsilon^3}}).
$$

\section{Proof of Proposition~\ref{prop:consistency-width}}
\label{sec:proof-consistency-width}

Using the context defined in
  Prop.~\ref{prop:consistency-width} and for $M$ satisfying \eqref{eq:prop-consist-width-minimal-cond}, we are going to show the contraposition of
  \eqref{eq:consistency-width}, \ie that with probability at least
  $1- 2 e^{- c \epsilon M / (1 + \delta)}$ for some $c>0$ and
  for all $\bs x,\bs y \in \cl K$ with $\bs x - \bs y \in
  \amgis_{K_0}$, having $\|\bs x - \bs y\| >
  \epsilon$ involves $\cl Q(\bs \Phi\bs x + \bs \xi) \neq \cl Q(\bs \Phi\bs
y + \bs \xi)$, or equivalently that 
\begin{equation}
  \label{eq:contrap-consist-width}
\|\bs x - \bs y\| >
  \epsilon \quad \Rightarrow\quad 
\Qdist(\bs x, \bs y) \geq
\tfrac{\delta}{M},
\end{equation}
from the definition of
$\Qdist$ in \eqref{eq:Qdist-decomp-indicator}.

The proof sketch is a follows. First, for some $\eta>0$, we create a
finite $\eta$-covering of
the set $\bar{\cl K} \subset \cl K \times \cl K$ of vector pairs whose difference belongs to
$\amgis_{K_0}$. Second, in order to show~\eqref{eq:contrap-consist-width}, we leverage the continuity of the
pseudo-distance $\Qdist^t$
under $\ell_2$-perturbations (Lemma~\ref{lem:continuity-L2-Dt}), as it happens that all points of $\bar{\cl K}$ are obtained by $\ell_2$-perturbations of the $\eta$-covering
and that, moreover, those perturbations are stable under projections by $\bs
\Phi$ (Lemma \ref{lem:small-resid-l2}). Finally, we adjust $\eta$ and some
additional parameters to show that, with
high probability, the softened distance $\Qdist^t(\bs
x_0, \bs y_0)$, for some $t$ depending on $\eta$, is large enough over
all pairs $(\bs x_0, \bs y_0)$ of the covering compatible with~$\|\bs
x - \bs y\| \geq \epsilon$, hence inducing \eqref{eq:contrap-consist-width}. 
\medskip

Let us define the set $\bar{\cl K} = \{(\bs x,\bs y) \in \cl K \times \cl
  K: \bs x - \bs y \in
  \amgis_{K_0}\} \subset \cl K \times \cl K$. We introduce a
  minimal $\eta$-net $\bar{\cl G}_\eta
  \subset \bar{\cl K}$ of
  $\bar{\cl K}$ with $0 < \eta < \epsilon/2$ to be specified later, such that for all $(\bs x,\bs y) \in \bar{\cl K}$,
  there exists a $(\bs x_0,\bs y_0) \in \bar{\cl G}_\eta$ with 
$$
\|(\bs x,\bs y) - (\bs x_0,\bs y_0)\| \leq \eta,
$$
which also involves $\|\bs x - \bs x_0\| \leq \eta$ and $\|\bs y - \bs
y_0\| \leq \eta$. 

The size of this minimal $\eta$-net is bounded as $\log |\bar{\cl G}_\eta| \leq 2 \cl H(\cl K,
\eta/\sqrt 2)$. Indeed, by the semi-additivity of the Kolmogorov entropy \cite[Theorem
2]{kolmogorov1959varepsilon}, $\bar{\cl K} \subset \cl K \times \cl K$
involves that $\cl H(\bar{\cl K}, \rho) \leq \cl
H(\cl K \times \cl K, \rho)$ for any $\rho > 0$. Since a $\rho$-net of
$\cl K \times \cl K$ can be obtained by the product $\cl G_{\rho'} \times
\cl G_{\rho'}$, with $\rho' = \rho/\sqrt 2$ and $\cl G_{\rho'}$ a
$\rho'$-net covering of $\cl K$, we
obtain  $\cl H(\bar{\cl K}, \rho) \leq 2 \cl H(\cl K, \rho/\sqrt
2)$. 

As for the proof of Prop.~\ref{prop:main-result}
in Sec.~\ref{sec:proofs}, by construction, all $(\bs x,\bs y) \in
\bar{\cl K}$ can also be written as 
$$
(\bs x,\bs y) = (\bs x_0,\bs y_0) + (\bs x', \bs y'),
$$
with $(\bs x_0, \bs y_0) \in \bar{\cl G}_\eta$, $(\bs x',\bs y') \in (\bar{\cl K} - \bar{\cl K}) \cap \eta\bb B^{2N}$. Notice that we
have also $\bs x',\bs y' \in \cl K_\eta := (\cl K - \cl K) \cap \eta
\bb B^N$, since $\bs x,\bs x_0, \bs y, \bs y_0 \in \cl K$ and $\max(\|\bs x'\|,\|\bs y'\|) \leq \|(\bs x',\bs
y')\|\leq \eta$.

As stated by Lemma~\ref{lem:small-resid-l2}, 
the diameter of the \emph{local set} $\cl K_\eta$ is stable with respect to random
projections. Since $\|\cl K_\eta\| \leq \eta$, there exist indeed two values ${C,c>0}$, only depending on the
sub-Gaussian norm $\alpha$, such that if 
\begin{equation}
  \label{eq:cond-M-small-resid-diam-bis}
  M \geq C \eta^{-2} w(\cl K_\eta)^2
\end{equation}
and $\bs \Phi \sim
\normsg^{M\times N}(0,1)$, we have with probability at least $1 - 2
\exp (- cM)$, 
\begin{equation}
\label{eq:small-resid-diam-bis}
\|\bs \Phi \cl K_\eta\| := \sup_{\bs u \in \cl K_\eta}
\|\bs \Phi \bs u\| \leq \sqrt M \|\cl K_\eta\|
\leq \eta \sqrt M.  
\end{equation}
Therefore, $\|\bs \Phi \bs x'\|  \leq \eta \sqrt
M$ and $\|\bs \Phi \bs y'\|  \leq \eta \sqrt M$ under the same conditions.
 
Moreover, if the previous event occurs, then, Lemma~\ref{lem:continuity-L2-Dt} for $t=0$ shows that for any $P\geq 1$,
\begin{equation}
\label{eq:final-inequality-to-prove}
\Qdist(\bs x, \bs y)  = \Qdist^0(\bs x_0 + \bs x', \bs y_0 + \bs y') \geq
\Qdist^{\eta \sqrt P}(\bs x_0, \bs y_0) - 4(\tfrac{\delta}{P} + \tfrac{\eta}{\sqrt P}).   
\end{equation}

Consequently, for reaching $\Qdist(\bs x,\bs y) \geq \delta/M$ as
expressed in \eqref{eq:contrap-consist-width}, since $\|\bs x-\bs y\| \geq
\epsilon$ involves $\|\bs
x_0-\bs y_0\| \geq \epsilon - 2\eta$, the proof can be deduced if we can
guarantee that, for all $(\bs u, \bs v) \in \bar{\cl G}_\eta$ with $\|\bs
u-\bs v\| \geq \epsilon - 2\eta$, the probability that
$\Qdist^{\eta\sqrt P}(\bs u, \bs v) \geq 4(\tfrac{\delta}{P} + \tfrac{\eta}{\sqrt P})
+ \frac{\delta}{M}$ tends (exponentially) to one with $M$.

Let us upper bound the corresponding probability of failure. We can
first observe the following result on a fixed pair of vectors. This
one is proved in App.~\ref{sec:proof-lemma-prob-bound-Qdist-gtr-r}. 
\begin{lemma}
  \label{lem:prob-bound-Qdist-gtr-r}
Let $\bs u, \bs v$ be in $\bb R^N$ with $\bs u - \bs v \in \amgis_{K_0}$
for some $K_0>0$ and $\|\bs u - \bs v\| \leq \epsilon_0$ for
$\epsilon_0 > 0$. For $\delta > 0$, $t\geq 0$, $r \in [\lfloor np \rfloor]$, $\bs\Phi \sim \normsg^{M\times N}(0,1)$, $\bs
\xi \sim \cl U^M([0,\delta])$ and the pseudo-distance $\Qdist^t$ defined
in \eqref{eq:Dt-def}, we have 
  \begin{equation}
    \label{eq:prob-bound-Qdist-gtr-r}
    \textstyle \bb P[\Qdist^{t}(\bs u, \bs v) \leq \frac{\delta}{M}  r] \leq \exp(-\frac{(Mp - r)^2}{2Mp}),
  \end{equation}
  with $p := \bb P\big[ d^t(\bs\varphi^\transp \bs u + \xi,
\bs\varphi^\transp \bs v + \xi) \neq 0 \big]$, $\bs \varphi \sim
\normsg^{N}(0,1)$ and $\xi \sim \cl U([0,\delta])$. Moreover, if $\sqrt K_0 \geq 16\kappasg$,
  \begin{equation}
    \label{eq:lower-bound-on-p}
    p  \geq \tfrac{1}{16(\delta +
  \epsilon_0)} \|\bs u - \bs v\| - \tfrac{2t}{\delta +
  \epsilon_0}.
  \end{equation}
\end{lemma}

From the discrete nature of $\Qdist^t$, the previous lemma (with
$t$ set to $\eta\sqrt P$) shows that for a fixed pair of vectors $\Qdist^{\eta\sqrt P}(\bs u, \bs v) \geq \frac{\delta}{M} (r +
1)$ holds with probability at least $1 - \exp(-(Mp -
r)^2/(2Mp))$. Moreover,~if
\begin{equation}
  \label{eq:cond-on-r}
  \tfrac{\delta}{M} r \geq 4 (\tfrac{\delta}{P} + \tfrac{\eta}{\sqrt P}),  
\end{equation}
we have
$$
\textstyle \Qdist^{\eta\sqrt P}(\bs u, \bs v) \geq \frac{\delta}{M} (r + 1)\ \Rightarrow\ \Qdist^{\eta\sqrt P}(\bs u, \bs v) \geq 4(\tfrac{\delta}{P} + \tfrac{\eta}{\sqrt P})
+ \frac{\delta}{M}.
$$
Therefore, setting $r=\lceil M p /2 \rceil \geq M p /2$,
\eqref{eq:prob-bound-Qdist-gtr-r} gives 
$$
\textstyle \bb P\big[\Qdist^{\eta\sqrt P}(\bs u, \bs v) \geq 4(\tfrac{\delta}{P} + \tfrac{\eta}{\sqrt P})
+ \frac{\delta}{M}\big] \geq 1 - \exp(-\frac{(Mp - r)^2}{2Mp}) > 1
- 2\exp(-\frac{Mp}{8}),
$$ 
if, from \eqref{eq:cond-on-r}, 
\begin{equation}
  \label{eq:cond-on-p}
  p  \geq \tfrac{8}{\delta} (\tfrac{\delta}{P} + \tfrac{\eta}{\sqrt P}).  
\end{equation}
Thus, we have to adjust $P$ and $\eta$
in order to satisfy \eqref{eq:cond-on-p}. Noting that $\epsilon - 2\eta
\leq \|\bs u - \bs v\|
\leq 2$ if $\cl K \subset \bb B^N$, \ie that we can set $\epsilon_0 = 2$ in
Lemma~\ref{lem:prob-bound-Qdist-gtr-r}, this adjustment
can be done from \eqref{eq:lower-bound-on-p} by imposing $B \geq C$~in 
\begin{equation}
  \label{eq:cond-on-p-bis}
  p\ \mathop{\geq}_{^{{\rm by\,}\eqref{eq:lower-bound-on-p}}}\ B := \tfrac{1}{16(\delta +
  2)} (\epsilon -
2\eta) - \tfrac{2\eta\sqrt P}{\delta +
  2}\ \geq\ C:= 8 (\tfrac{1}{P} + \tfrac{\eta}{\delta\sqrt P}).  
\end{equation}
A solution is to set, for
some $c\geq 1$ and $d>0$ to be specified later,
$P = c^2 \tfrac{2 + \delta}{\epsilon} \geq 1$ and $\eta = d \tfrac{\epsilon^{3/2}}{\sqrt{2 + \delta}} \leq d \epsilon$.
Then
$$
\epsilon - 2\eta \geq (1 - 2d)\epsilon,\quad
\eta \sqrt P = cd\,\epsilon,\quad
\tinv{P} = \tfrac{1}{c^2(2 + \delta)}\epsilon, \quad
\tfrac{\eta}{\delta \sqrt P} = \tfrac{d}{c}
\tfrac{\epsilon^{2}}{\delta (2 + \delta)} \leq \tfrac{d}{c}
\tfrac{2}{\delta (2 + \delta)}\,\epsilon,
$$
so that
$$
B \geq
\tfrac{1 -
2d - 32 cd}{16(\delta +
  2)}\,\epsilon, \quad
C \leq \tfrac{8}{c^2(2 + \delta)} (1 + cd
\tfrac{2}{\delta})\,\epsilon.
$$
Fixing $d = \tinv{2} (32)^{-2}\tfrac{\delta}{\delta + 2} <
\tinv{2} (32)^{-2}$ and $c = 32$, a few estimations show finally that 
$$
\epsilon^{-1} B \geq \tfrac{1 - (32)^{-2} - \inv{2}}{16(\delta +
  2)} \geq \tinv{33(\delta + 2)}, \quad
\epsilon^{-1} C \leq \tfrac{8}{(32)^2(2 + \delta)}(1 +
\tfrac{2}{(64)(\delta + 2)}) < \tinv{64(\delta +
  2)},
$$
proving that for our choice of parameters, \ie for $P = (32)^2\,
\tfrac{2 + \delta}{\epsilon} \geq 1$ and $\eta = \tinv{2} (32)^{-2}
\delta (\tfrac{\epsilon}{2 + \delta})^{3/2}$,
\eqref{eq:cond-on-p} can be satisfied since $B \geq C$. Moreover, for this choice of
parameters, \eqref{eq:cond-on-p} provides 
$$
p \geq \tfrac{\epsilon}{33(2 + \delta)}.
$$

We are now ready to complete the proof. Using the previous
developments, defining $\bar{\cl G}'_\eta := \{(\bs u,\bs v) \in \bar{\cl G}_\eta: \|\bs u -
\bs v\| \geq \epsilon - 2\eta\} \subset \bar{\cl G}_\eta$ with $\eta \simeq \delta
\epsilon^{3/2} (2 + \delta)^{-3/2}$ fixed as above and $\log |\bar{\cl G}'_\eta| \leq \log |\bar{\cl G}_\eta| \leq 2 \cl
H(\cl K, \eta/\sqrt 2)$ as explained before, by a simple union bound there exist some constants $C,c,c'>0$ such that if 
$$
M \geq C \tfrac{2 + \delta}{\epsilon}\,\cl H\big(\cl K,
c\,\delta (\tfrac{\epsilon}{2 + \delta})^{3/2}\big),  
$$
then the event 
\begin{equation}
  \label{eq:lower-bouded-Qdist-all-pairs-in-cover}
  \textstyle \Qdist^{\eta\sqrt P}(\bs u, \bs v) \geq 4(\tfrac{\delta}{P} + \tfrac{\eta}{\sqrt P})
  + \frac{\delta}{M},\quad \forall \bs u, \bs v \in \bar{\cl
    G}'_\eta,  
\end{equation}
holds with probability at least 
$$
\textstyle 1 - 2\exp(\,2 \cl H(\cl K, \tfrac{\eta}{\sqrt 2}) -
\tfrac{Mp}{8})\ \geq\ 1 - 2\exp(\,2 \cl H(\cl K, \tfrac{\eta}{\sqrt 2}) -
\tfrac{M\epsilon}{33(2 + \delta)}) \ \geq\ 1 - 2 \exp(-c'\,\tfrac{M\epsilon}{2 + \delta}).
$$

Remembering that for having \eqref{eq:final-inequality-to-prove} the diameter of $\cl
K_\eta$ must remain small under random projections by $\bs \Phi$ (as stated in \eqref{eq:small-resid-diam-bis}), so by
imposing \eqref{eq:cond-M-small-resid-diam-bis}, we find again by
union bound that for some
other constants $C,c,c'>0$, if 
\begin{equation}
  \label{eq:involved-cond-on-M}
  M \geq C \max\bigg( \tfrac{(2 + \delta)^3}{\delta^2\epsilon^3} w(\cl K_{c\,\delta (\tfrac{\epsilon}{2 + \delta})^{3/2}})^2, \tfrac{2 + \delta}{\epsilon}\,\cl H\big(\cl K, c\,\delta(\tfrac{\epsilon}{2 + \delta})^{3/2}\big)\bigg),    
\end{equation}
then, with probability at least $1 - 4\exp(- c' M\epsilon/(2
+ \delta))$, for all $\bs x,\bs y \in \cl K$ with $\bs x - \bs y \in
\amgis_{K_0}$ and $\|\bs x - \bs y\| \geq \epsilon$, \eqref{eq:final-inequality-to-prove} combined with
\eqref{eq:lower-bouded-Qdist-all-pairs-in-cover} provides
$$
\Qdist(\bs x,\bs y) \geq \tfrac{\delta}{M},
$$
as requested at the beginning.

We conclude the proof by simplifying the general condition
\eqref{eq:involved-cond-on-M}.
First, for a general bounded set $\cl K$, Sudakov inequality
(P.~\ref{p:w-suda}) and
Sec.~\ref{sec:proofs} provide $\cl H(\cl K, \eta) \lesssim \tfrac{w(\cl K)^2}{\eta^2}$ and
$w(\cl K_\eta) \leq 2 w(\cl K)$, so that \eqref{eq:involved-cond-on-M}
holds if   
$$
M \geq C \tfrac{(2 + \delta)^4}{\delta^2 \epsilon^4} \, w(\cl K)^2,    
$$
for another constant $C>0$.

Second, if the set $\cl K$ is structured,  then,
  from~\eqref{eq:structured-set-prop} and the same
  simplifications used for Prop.~\ref{prop:main-relaxed-result} to
  reach Prop.~\ref{prop:main-result}, the right-hand
  side of~\eqref{eq:involved-cond-on-M} can be bounded~by
  \begin{align*}
\max\big( (\delta s)^{-2} w(\cl K_{c\,\delta s})^2, s^{-2/3}\,\cl
H\big(\cl K, c\,\delta s)\big)&\leq 
\max\big( c^2 \bar w(\cl K)^2, \tfrac{2+\delta}{\epsilon} 
\bar w(\cl K)^2 \log(1+\tfrac{\|\cl K\|}{c\,\delta s})\big)\\
&\lesssim \tfrac{2+\delta}{\epsilon}\, 
\bar w(\cl K)^2 \log\big(1+\tfrac{(2 + \delta)^{3/2}\|\cl K\|}{\delta \epsilon^{3/2}}\big),  
  \end{align*}
with $s := \epsilon^{3/2} / (2 + \delta)^{3/2}$, which explains the
requirement \eqref{eq:prop-consist-width-minimal-cond-K-sparse}. 

\section{Acknowledgements}
\label{sec:acknowledgements}

We wish to gladly thank Holger Rauhut and Sjoerd Dirksen for
interesting and enlightening discussion on quantized random projections
during a short stay end of January 2015 in RWTH Aachen University, and Jerry Veeh (Auburn University, AL, USA) for interesting discussions
on the error bounds of the Stirling's approximation, as deduced in his
lecture notes \cite{jveeh_stirling}, and for having
pointed out the work \cite{mortici11}. We also thank Valerio Cambareri
(UCLouvain, Belgium) for interesting discussions on quantized
embeddings and for his advices on the writing of this~paper. 

\appendix

\section{On the absolute expectation of a difference of dithered floors}
\label{sec:absol-expect-dith}

This short appendix proves the equality
$$
\bb E |\lfloor x + \xi \rfloor - \lfloor y + \xi \rfloor| = |x -
y|,\quad \forall x,y \in \bb R,\ \xi \sim \cl U([0,1]).
$$
Denoting $a = \lfloor x \rfloor \in \bb Z$, $b = \lfloor y \rfloor \in \bb Z$, $x' = x -
a \in [0,1)$ and $y' = y - b \in [0,1)$, since $\lfloor \lambda - n
\rfloor = \lfloor \lambda \rfloor - n
$ for any $\lambda \in \bb R$ and $n \in \bb Z$, we can always write 
\begin{align*}
\bb E |\lfloor x + \xi \rfloor - \lfloor y + \xi \rfloor|&=\bb E |a -
                                                           b + X|,  
\end{align*}
with $X = \lfloor x' + \xi \rfloor - \lfloor y' + \xi
\rfloor$. Without loss of generality, we can assume that the \rv $X$
is positive, \ie $x' \geq y'$ (just flip the role of $x$ and $y$ if this is not the
case). Moreover, since $x',y'\in [0,1)$, $X \in \{0,1\}$ and 
\begin{align*}
\bb P(X = 0)&=\ \bb P(x' + \xi < 1, y' + \xi < 1) + \bb P(x' + \xi \geq
1, y' + \xi \geq 1)\\
&=\ \bb P(x' + \xi < 1) + \bb P(y' + \xi \geq 1) = 1
- x' + y'.
\end{align*}
Therefore, 
\begin{align}
\bb E |a - b + X|&= (|a - b| - |a - b + 1|)\,\bb P(X=0) + |a - b +
                   1|\nonumber\\
\label{eq:tmp1}
&= |a - b| - (x' - y') (|a - b| - |a - b + 1|). 
\end{align}
If $x' = y'$, then $\bb E |a - b + X| = |a-b| = |x-y|$. Let us
consider now the case $x' > y'$. If $x - y \geq 0$, then $a - b \geq y' - x' > -1$ since $x' < 1$, \ie
$a - b \geq 0$ since $a - b \in \bb Z$. Consequently,~\eqref{eq:tmp1}
provides $\bb E |a - b + X| = a - b + x' - y' = x - y$. When $x - y <
0$, $b - a > x' - y' > 0$, \ie $a - b \leq a - b + 1 \leq 0$, 
and we get $\bb E |a - b + X| = b - a - (x' - y') = x - y$. In summary,  $\bb E |a - b +
X| = |x - y|$ in all cases, which proves the result.

\section{Proof of Lemma~\ref{lem:bounded-dbar}}
\label{sec:proof-lemma-bounded-dbar}

We start by observing that
\begin{align*}
\textstyle \tinv{\delta}\,\big|d^t(a, b) - d^s(a, b)\big|&\textstyle \leq \sum_{k\in\bb Z} \big|\bbone[{\cl
      F^t(a - k\delta,b - k\delta)}] - \bbone[{\cl
      F^s(a - k\delta,b - k\delta)}]\big|\\
    &\textstyle \leq\sum_{k\in\bb Z} \bbone[{\cl
      H^{t,s}(a - k\delta,b - k\delta)}]
\end{align*}
with 
$$
\cl H^{t,s}(a,b)\ :=\ \cl F^t(a,b)\ \triangle\ \cl F^s(a,b)\ :=\ \big (\cl F^t(a,b) \cup
\cl F^s(a,b) \big) \setminus \big (\cl F^t(a,b) \cap \cl F^s(a,b)\big).
$$
For $t\geq s$, $\cl F^t(a,b) \subset \cl F^s(a,b)$ and $\cl
H^{t,s}(a,b) = \cl F^s(a,b) \setminus \cl F^t(a,b)$, while for $t <
s$, $\cl H^{t,s}(a,b) = \cl F^t(a,b) \setminus \cl F^s(a,b)$.
Moreover, a careful piecewise analysis made on the different
sign combinations for $s$
and $t$ show that $\cl H^t(a,b)
\subset \{|a| \in [r_-,r_+]\} \cup \{|b| \in [r_-,r_+]\}$ with $r_+ :=
\max(|s|,|t|)$ and $r_-$  
equals to $\min(|s|,|t|)$ if
$ts\geq 0$ and 0 otherwise. Consequently,
writing $r = r_+ - r_- \leq |t - s|$,
  \begin{align*}
    \big|d^t(a, b) - d^s(a,
    b)\big|&\textstyle\ \leq\ \delta\,\sum_{k\in\bb Z} \bbone\big[\{|a-k\delta| \in [r_-,r_+]\} \cup
    \{|b-k\delta|\in [r_-,r_+]\}\big]\\
&\leq\ 2 \delta ( \tfrac{2r}{\delta} + 2) = 4 (|t-s| + \delta).
  \end{align*}
Moreover, if $s=0$, since then $r_-=0$ and $r_+=r = |t|$,  
$$
\textstyle \sum_{k\in\bb Z} \bbone\big[\{|a-k\delta| \leq |t|\} \cup
    \{|b-k\delta|\leq |t|\}\big] \leq 2 \delta ( \tfrac{2|t|}{\delta}
    + 1) = 4|t| + 2\delta,  
$$
and we find
\begin{align*}
\big|d^t(a, b) - |a-b|\big|&\leq\ \big|d^t(a, b) -
d(a, b) \big| + \big|d(a, b) -|a
- b|\big|\\
&=\ \big|d^t(a, b) -
d^0(a, b) \big| + \big| |\cl Q(a) - \cl Q(b)| -|a
- b|\big|\\
&\leq\ (4 |t| + 2\delta) + 2 \delta\ = 4\,(|t| + \delta).
\end{align*}

\section{Proof of Lemma~\ref{lem:subGaussian-Dt}}
\label{sec:proof-lemma-subGaussian-Dt}

Let us define $\tilde Z := |\bs\varphi^\transp(\bs x-\bs y)| = |a-b|$
with the two \rv's $a=\bs\varphi^\transp\bs x +
\xi$ and $b=\bs\varphi^\transp\bs x +
\xi$. From \eqref{eq:law-of-total-expectation-on-Qdist}, 
$\bb E \tilde Z = \bb E Z^0$. Moreover, from the approximate
rotational invariance property \eqref{eq:approx-rot-inv}, $\tilde Z$ is
sub-Gaussian with $\|\tilde Z\|_{\psi_2} = \|\bs\varphi^\transp(\bs x-\bs y)\|_{\psi_2}
\lesssim \|\bs x - \bs y\|$, and using Lemma~\ref{lem:bounded-dbar}
and the bound $\|\cdot\|_{\psi_2} \leq \|\cdot\|_{\infty}$, we find
  \begin{align*}
\|Z^t\|_{\psi_2}&\leq \|Z^t - \tilde Z\|_{\psi_2} + \|\tilde Z\|_{\psi_2}\\
&\lesssim \|d^t(a, b) - |a - b|\|_{\psi_2} + \|\bs x
- \bs y\|\\
&\lesssim \delta + |t| + \|\bs x
- \bs y\|,
  \end{align*}
which demonstrates the sub-Gaussianity of $Z^t$. 

For the expectation,
writing $a=a'+\xi$ and $b=b'+\xi$ with $a'=\bs\varphi^\transp\bs x$ and $b'=\bs\varphi^\transp\bs y$,
by Jensen's inequality and the law of total expectation, we find
\begin{align*}
\big|\bb E Z^t - E Z^0 \big|&\leq \bb E|Z^t - Z^0| = \bb E_{\bs
  \varphi} \bb E_\xi |d^t(a'+\xi, b'+\xi) - d(a'+\xi, b'+\xi)|.
\end{align*}
However, reusing some elements of
the proof of Lemma~\ref{lem:bounded-dbar} and considering $\bs
\varphi$ fixed,
\begin{align*}
&\textstyle \bb E_\xi |d^t(a'+\xi, b'+\xi) - d(a'+\xi, b'+\xi)|\big|\\
&\textstyle\leq \delta\,\sum_{k\in\bb Z} \bb E_\xi \bbone\big[\{|a'+\xi-k\delta|\leq |t|\} \cup
    \{|b'+\xi-k\delta|\leq |t|\}\big]\\
&\textstyle\leq \delta\,\sum_{k\in\bb Z}
\bb E_\xi\bbone\big[\{|a'+\xi-k\delta|\leq |t|\}\big] + \delta\,\sum_{k\in\bb Z}
\bb E_\xi\bbone\big[\{|b'+\xi-k\delta|\leq |t|\}\big].  
\end{align*}
Moreover, since $\xi \sim \cl U([0,\delta])$,
\begin{align*}
\textstyle\delta\,\sum_{k\in\bb Z}
\bb E_\xi\bbone\big[\{|a'+\xi-k\delta|\leq
|t|\}\big]&\textstyle=\sum_{k\in\bb Z}\int_0^\delta\bbone\big[\{|a'+s-k\delta|\leq
|t|\}\big]\,\ud s\\
&=\textstyle\int_{\bb R} \bbone\big[\{|a'+s|\leq
|t|\}\big] \ud s\ =\ 2|t|,  
\end{align*}
which provides also $\delta\,\sum_{k\in\bb Z}
\bb E_\xi\bbone\big[\{|b'+\xi-k\delta|\leq |t|\}\big] =
2|t|$. Consequently,  since these two quantities do not depend on
$\bs \varphi$, we find $\big|\bb E Z^t - E Z^0 \big| \lesssim
|t|$. Finally, if $\bs \varphi \sim \cl N^{N}(0,1)$, $Z^0 \sim \cl
N(0,\|\bs x-\bs y\|^2)$, and $\bb E |Z^0| = \gone \|\bs x -\bs y\|$.

\section{Proof of Lemma~\ref{lem:continuity-L2-Dt}}
\label{sec:app-proof-cont-L2-perturb}

We adapt the proof of Lemma 5.5 in
\cite{plan2011dimension} to both $\ell_2$-perturbations (instead of
$\ell_1$ ones) of $\bs x_0$ and
$\bs y_0$, and to the context of uniform dithered quantization instead
of 1-bit (sign) quantization. By assumption, we have $\|\bs \Phi \bs x'\| \leq \eta \sqrt M$ and $\|\bs \Phi \bs y'\|
  \leq \eta \sqrt M$. Therefore, the set 
$$
T := \{i \in [M]: |(\bs \Phi \bs x')_i| \leq \eta \sqrt P, |(\bs \Phi
\bs y')_i| \leq \eta \sqrt P\}
$$
is such that $|T^{\compl}| \leq 2 M/P$ as $2\eta^2 M \geq \|\bs \Phi
\bs x'\|^2 + \|\bs \Phi \bs y'\|^2 \geq \|(\bs \Phi
\bs x')_{T}\|^2 + \|(\bs \Phi \bs y')_{T}\|^2 + |T^\compl|P\eta^2 \geq
|T^\compl|P\eta^2$. Considering the definition of $\cl F^t$ in
\eqref{eq:def-Ft}, we have, for all $i \in T$ and any $\lambda \in \bb
R$,
$$
\cl F_i^{t + \eta\sqrt{P}}(\bs x_0, \bs y_0, \lambda) \subset 
\cl F_i^{t}(\bs x_0 + \bs x', \bs y_0 + \bs y', \lambda) \subset 
\cl F_i^{t - \eta\sqrt{P}}(\bs x_0, \bs y_0, \lambda),
$$
with $\cl F^t_i(\bs x_0,\bs y_0, \lambda) := \cl F^{t}(\bs\varphi_i^\transp\bs x_0
+ \xi_i - \lambda,\bs\varphi_i^\transp\bs
y_0 + \xi_i - \lambda)$.

Denoting $a_i =
\max(|\bs\varphi_i^\transp\bs x'|,|\bs\varphi_i^\transp\bs y'|)$, we find 
\begin{align*}
&\Qdist^{t+\eta\sqrt P}(\bs x_0,\bs y_0) =\textstyle \tfrac{\delta}{M} \sum_{i=1}^M \sum_{k\in\bb Z} \bbone[\cl
F_i^{t+\eta\sqrt P}(\bs x_0,\bs y_0, k\delta)]\\
&\leq \textstyle \tfrac{\delta}{M} \sum_{i\in T} \sum_{k\in\bb Z} \bbone[\cl
F_i^{t}(\bs x_0 + \bs x',\bs y_0 + \bs y', k\delta)] + \tfrac{\delta}{M} \sum_{i\in T^\compl} \sum_{k\in\bb Z} \bbone[\cl
F_i^{t+\eta\sqrt P-a_i}(\bs x_0 + \bs x',\bs y_0 + \bs y', k\delta)]\\  
&\leq \textstyle \tfrac{\delta}{M} \sum_{i\in T} \sum_{k\in\bb Z} \bbone[\cl
F_i^{t}(\bs x_0 + \bs x',\bs y_0 + \bs y', k\delta)] + \tfrac{\delta}{M} \sum_{i\in T^\compl} \sum_{k\in\bb Z} \bbone[\cl
F_i^{t}(\bs x_0 + \bs x',\bs y_0 + \bs y', k\delta)]\\  
&\qquad\textstyle + \tfrac{1}{M} \sum_{i\in T^\compl} \delta\sum_{k\in\bb Z} \big|\bbone[\cl
F_i^{t+\eta\sqrt P-a_i}(\bs x_0 + \bs x',\bs y_0 + \bs y', k\delta)] - \bbone[\cl
F_i^{t}(\bs x_0 + \bs x',\bs y_0 + \bs y', k\delta)]\big|.
\end{align*}
Using \eqref{eq:diff-dt-ds} to bound the last sum of the last
expression and since, by definition of $T$, $a_i \geq \eta\sqrt P$ for $i\in T^\compl$, we find 
\begin{align*}
\Qdist^{t+\eta\sqrt P}(\bs x_0,\bs y_0)&\textstyle \leq \Qdist^{t}(\bs
                                       x_0 + \bs x',\bs y_0 + \bs y') +
                                       \tfrac{4}{M} \sum_{i\in
                                       T^\compl} (\delta + a_i - \eta\sqrt P)\\
&\textstyle \leq \Qdist^{t}(\bs x_0 + \bs x',\bs y_0 + \bs y') +
\tfrac{4\delta}{P} + \tfrac{4}{M}\sum_{i\in T^\compl} (a_i -
  \eta \sqrt P) \\
&\textstyle \leq \Qdist^{t}(\bs x_0 + \bs x',\bs y_0 + \bs y') +
\tfrac{4\delta}{P} + \tfrac{4}{M}\sum_{i\in T^\compl} a_i -
  \tfrac{4 |T^\compl|}{M}\eta \sqrt P.
\end{align*}
However, 
$$
\textstyle \tfrac{1}{M}\sum_{i\in T^\compl} a_i \leq \tfrac{1}{M}(\|(\bs \Phi \bs
x')_{T^\compl}\|_1 + \|(\bs \Phi \bs
y')_{T^\compl}\|_1) \leq \tfrac{\sqrt{|T^{\compl}|}}{M}(\|(\bs \Phi \bs
x')_{T^\compl}\| + \|(\bs \Phi \bs
y')_{T^\compl}\|) \leq 2 \eta \sqrt{\tfrac{|T^\compl|}{M}},
$$
and since $f(t)=2t-t^2\sqrt P \leq 1/\sqrt P$ for all $t \in \bb R$,
we find
\begin{align*}
\Qdist^{t+\eta\sqrt P}(\bs x_0,\bs y_0)&\textstyle \leq \Qdist^{t}(\bs x_0 + \bs x',\bs y_0 + \bs y') +
\tfrac{4\delta}{P} + 4\eta\,( 2 \sqrt{\tfrac{|T^\compl|}{M}} -
  \tfrac{|T^\compl|}{M} \sqrt P)\\
&\textstyle \leq \Qdist^{t}(\bs x_0 + \bs x',\bs y_0 + \bs y') +
\tfrac{4\delta}{P} + 4\tfrac{\eta}{\sqrt P},
\end{align*}
which provides the lower bound of \eqref{eq:continuity-L1-Dt}.

For the upper bound, 
\begin{align*}
&\Qdist^{t-\eta\sqrt P}(\bs x_0,\bs y_0) =\textstyle \tfrac{\delta}{M} \sum_{i=1}^M \sum_{k\in\bb Z} \bbone[\cl
F_i^{t-\eta\sqrt P}(\bs x_0,\bs y_0, k\delta)]\\
&\geq \textstyle \tfrac{\delta}{M} \sum_{i\in T} \sum_{k\in\bb Z} \bbone[\cl
F_i^{t}(\bs x_0 + \bs x',\bs y_0 + \bs y', k\delta)] + \tfrac{\delta}{M} \sum_{i\in T^\compl} \sum_{k\in\bb Z} \bbone[\cl
F_i^{t-\eta\sqrt P+a_i}(\bs x_0 + \bs x',\bs y_0 + \bs y', k\delta)]\\  
&\geq \textstyle \Qdist^{t}(\bs x_0 + \bs x',\bs y_0 + \bs y')\\  
&\qquad\textstyle - \tfrac{1}{M} \sum_{i\in T^\compl} \delta\sum_{k\in\bb Z} \big|\bbone[\cl
F_i^{t}(\bs x_0 + \bs x',\bs y_0 + \bs y', k\delta)] - \bbone[\cl
F_i^{t-\eta\sqrt P+a_i}(\bs x_0 + \bs x',\bs y_0 + \bs y', k\delta)]\big|,
\end{align*}
and, as above, the last sum can be upper-bounded by
$\tfrac{4\delta}{P} + \tfrac{4\eta}{\sqrt P}$ using \eqref{eq:diff-dt-ds}.

\section{Proof of Lemma~\ref{lem:small-resid-l2}}
\label{sec:proof-prop-small-resid-l2}

We use here a similar proposition of Mendelson\footnote{Where a totally
    equivalent sub-Gaussian norm is used, \ie $\|X\|^{(\text{Mend.})}_{\psi_2} :=
    \inf\{s: \bb E \exp( X^2/s^2) \leq 2\}$ with
    $\|X\|^{(\text{Mend.})}_{\psi_2} \simeq \|X\|_{\psi_2}$
    \cite{IntroNonAsRandom}.}  \emph{et al.} in \cite{mendelson2008uniform} for subsets of
  $\bb S^{N-1}$ that we lift to subsets of $\bb R^{N+1}$ thank to  some tools developed in
  \cite{plan2011dimension} for other purposes.

We fix $t= \|\cl R\|/\sqrt 6$ and form the set $\cl R' := \{\bs u/\|\bs u\|: \bs u \in {\cl R \oplus
t}\}$ 
with $\cl R \oplus t := \{
(\begin{smallmatrix}
\bs x\\
t  
\end{smallmatrix})
: \bs x\in \cl R\}\, \subset \bb
R^{N+1}$. As $\cl R' \subset \bb S^N$, we know from \cite[Theorem 2.1]{mendelson2008uniform}
that for $0<\epsilon<1$,
$$
M\ \gtrsim\ \tfrac{\alpha^4}{\epsilon^2} w(\cl R')^2 
$$ 
and $\bs \Phi' \sim \normsg^{M\times (N+1)}(0,1)$, 
$$
\textstyle \bb P\big[\sup_{\bs x'\in\cl R'} |\tinv{M}\|\bs \Phi'\bs x'\|^2 - 1| \geq
\epsilon\big] \leq \exp(-c \tfrac{\epsilon^2 M}{\alpha^4}).
$$
However, for $\bs g \sim \cl
N^{N}(0,1)$ and $\gamma \sim \normpdf(0,1)$, as observed similarly in \cite{plan2011dimension}, 
\begin{align*}
  w(\cl R') &= \bb E \sup_{\bs x \in \cl R} (\|\bs x\|^2 + t^2)^{-1/2}
              |\scp{\bs g}{\bs x} + t\gamma|
              \ \leq \tinv{t}\,(\bb E
              \sup_{\bs x \in \cl R} |\scp{\bs g}{\bs x}| + t
  \gone)\\
&\leq \tfrac{\sqrt 6}{\|\cl R\|} w(\cl R) + \gone\ \leq\ 4 \tfrac{w(\cl R)}{\|\cl R\|},
\end{align*}
since, for all $\bs x \in \cl R$, $w(\cl R) \geq \gone \|\bs x\|$, \ie
$w(\cl R) \geq \gone \|\cl R\|$.  Therefore, fixing $\epsilon = 1/2$, if
$M \gtrsim \alpha^4 w(\cl R)^2/\|\cl
K\|^2$, with probability at least $1- e^{-c\,\alpha^{-4}M}$,  we have, for all $\bs
x \in \cl R$, 
\begin{align*}
\sqrt{\tfrac{3}{2}} \geq \tinv{\sqrt{M}}\|\bs \Phi'\bs
  x'\|&\geq \tinv{t\sqrt{M}}\|\bs \Phi\bs x + t\bs\phi\|\ \geq
        \tfrac{1}{t\sqrt M} (\|\bs \Phi \bs x\| - t\|\bs
        \Phi'
        (\begin{smallmatrix}
          \bs 0\\
          1
        \end{smallmatrix})
\|)\
\geq
        \tfrac{1}{t\sqrt M} \|\bs \Phi \bs x\| - \sqrt{\tfrac{3}{2}},
\end{align*}
where $\bs x' = \|(\begin{smallmatrix}
\bs x\\
t  
\end{smallmatrix})\|^{-1} \, (\begin{smallmatrix}
\bs x\\
t  
\end{smallmatrix}) \in \cl R'$, $\bs \phi \in \bb R^M$ is the last column of $\bs \Phi'$ and using the fact that $(\begin{smallmatrix}
\bs 0\\
1  
\end{smallmatrix}) \in \cl R'$ since $0 \in \cl R$. Therefore, replacing
$t$ by its value, we find with the same probability,
$$
\tfrac{1}{\sqrt M} \|\bs \Phi \bs x\| \leq \|\cl R\|,
$$
for all $\bs x \in \cl R$, \ie $\|\bs \Phi \cl R\| \leq \sqrt M \|\cl R\|$.

\section{Proof of Lemma~\ref{lem:prob-bound-Qdist-gtr-r}}
\label{sec:proof-lemma-prob-bound-Qdist-gtr-r}

From the relation $\textstyle\Qdist^t(\bs u,\bs v) = \tfrac{1}{M} \sum_{i=1}^M d^t(\rmap_i(\bs u),
\rmap_i(\bs v))$ established in Sec.~\ref{sec:proofs} between $\Qdist^t$
and $d^t \in \delta \bb N$ defined in \eqref{eq:barD-def}, and associated to the
vectorial mapping $\bs u\in \bb
R^N \to \brmap(\bs u) = \bs \Phi \bs u + \bs \xi$ whose components are
independent, we reach the bound \eqref{eq:prob-bound-Qdist-gtr-r} with
the cdf of a binomial distribution:  since
\begin{align*}
\textstyle\bb P\big[\frac{M}{\delta} \Qdist^t(\bs u, \bs v) \leq
  r\big]\leq&\textstyle\ \bb P\big[\,\big|\{j \in [M]:\,  d^t(\rmap_i(\bs u),
\rmap_i(\bs v)) \neq 0\}\big| \leq r\big]\\
=&\textstyle\ \sum_{k=0}^{r} {M\choose k} p^k (1-p)^{M-k},
\end{align*}
Chernoff's inequality can upper bound this binomial cdf with
\begin{equation}
\label{eq:tail-binom}
\textstyle\bb P\big[\frac{M}{\delta} \Qdist^t(\bs u, \bs v) \leq
  r\big] \leq \exp(-\frac{(Mp - r)^2}{2Mp}).
\end{equation}

Let us now lower bound $p$. Defining $\bs w=\bs u - \bs v \in \amgis_{K_0}$ and $\hat{\bs w} = \bs
w/\|\bs w\|$,
the action of dithering $\xi \sim \cl U([0, \delta])$
allows us to 
compute easily that, 
$$
p = \bb E_{\bs \varphi}\,\bb P_\xi\big[ d^t(\bs\varphi^\transp \bs u + \xi,
\bs\varphi^\transp \bs v + \xi) \neq 0 \big] = \bb E \min\big(1,
\delta^{-1} (|\bs\varphi^\transp {\bs w}|-2t)_+\big).
$$
In order to avoid any further singularity when $\delta \to 0$, we can benefit from
the fact that $p\geq 1$ and
work with this slightly looser bound:  
$$
p \geq \bb E \min\big(1,
(\epsilon_0+\delta)^{-1} (|\bs\varphi^\transp {\bs w}|-2t)_+\big).
$$
Moreover, with $\alpha = \|\bs u - \bs v\|/(\delta+\epsilon_0)$, 
$$
p  \geq \textstyle \bb E \min(1,
\alpha |\bs\varphi^\transp \hat{\bs w}|-\tfrac{2t}{\delta+\epsilon_0}) \geq \bb E \min(1,
\alpha |\bs\varphi^\transp \hat{\bs w}|) - \tfrac{2t}{\delta+\epsilon_0},
$$
so that
\begin{equation}
  \label{eq:p-low-bound}
\textstyle p\ \geq\ \bb E \min(1,
\alpha |g|) - \tfrac{2t}{\delta+\epsilon_0} - A,  
\end{equation}
where $g\sim \cl N(0,1)$ and $A:=|\bb E \min(1,
\alpha |\bs\varphi^\transp \hat{\bs w}|)-\bb E \min(1,
\alpha |g|)|$.

We can upper bound $A$ from our assumptions on the sub-Gaussian vector
$\bs \varphi \sim \normsg^N(0,1)$:
\begin{align*}
A&\textstyle =\ \big|\int_0^1 \bb P(\min(1,
\alpha |\bs\varphi^\transp \hat{\bs w}|) \geq u) - \bb P(\min(1,
\alpha |g|) \geq u) \,\ud u \,\big|\\
&\textstyle =\ \big|\int_0^1 \bb P(\alpha\,|\bs\varphi^\transp \hat{\bs w}|
  \geq u) - \bb P(\alpha |g| \geq u) \,\ud u\,\big|\\  
&\textstyle \leq \alpha \int_0^{+\infty} \big|\bb P(|\bs\varphi^\transp \hat{\bs w}|
  \geq u) - \bb P(|g| \geq u)\big| \,\ud u\\  
&\textstyle \leq \tfrac{\kappasg}{\delta + \epsilon_0}\,\|\bs w\|_{\infty}\ \leq\ \tfrac{\kappasg}{\sqrt K_0}\,\alpha,
\end{align*}
where the last inequalities rely on assumption
\eqref{eq:Berry-Esseen-relation} (setting $\bs u = \hat{\bs w}$) and on the fact that $\bs w \in \amgis_{K_0}$.

Moreover, for lower-bounding $\bb E \min(1,
\alpha |g|)$ in \eqref{eq:p-low-bound}, we observe that $\min(1,\alpha
x) = \alpha x
- \alpha (x - 1/\alpha)_+$ for $x\in \bb R$. Therefore, defining $F(x) :=
\tinv{2} \alpha x^2 - \tinv{2}\alpha (x - 1/\alpha)^2_+ = \int_0^x
\min(1,\alpha u)\ud u$ and integrating by parts, we find
$$
\textstyle \bb E \min(1,
\alpha |g|) = \bb E(|g| F(|g|)) \geq \gone F(\gone)
$$
where in the last inequality we used Jensen's inequality and the convexity of $x\in\Rbb_+ \mapsto x
F(x)$. It is easy to see that $2F(x) \geq
\alpha x^2/(1 + \alpha x)$ so that
$$
\textstyle \bb E \min(1,
\alpha |g|) \geq \tinv{2} \gone \tfrac{\frac{2}{\pi}\alpha }{1 + \gone
  \alpha } \geq \tinv{4} \tfrac{\alpha}{1 + \alpha}.
$$
Finally,
$$
p \geq \tinv{4} \tfrac{\alpha}{1 + \alpha} - \tfrac{2t}{\delta+\epsilon_0} -
\tfrac{\alpha\kappasg}{\sqrt K_0} \geq \tinv{4} \tfrac{1}{\delta +
  2\epsilon_0} \|\bs u - \bs v\| - \tfrac{2t}{\delta+\epsilon_0} - \tfrac{\alpha\kappasg}{\sqrt K_0} \geq
\tfrac{1}{\delta +
  \epsilon_0} (\tinv{8} - \tfrac{\kappasg}{\sqrt K_0}) \|\bs u - \bs v\| - \tfrac{2t}{\delta +
  \epsilon_0},
$$  
the last expression providing \eqref{eq:lower-bound-on-p} if $\sqrt
K_0 \geq 16 \kappasg$.

\section{A lower bound on the approximation error of the Mean
  Absolute Difference of a binomial random variable}
\label{app:bin-bound}

This small section establishes a lower bound on the approximation
error of the MAD $M_n := \bb E |\beta_{n} - \bb E \beta_{n}|$ of a binomial random variable $\beta_n \sim {\rm
  Bin}(n,1/2)$ by a fraction of its standard deviation $\sigma_n := (\bb E
|\beta_{n} - \bb E \beta_{n}|^2)^{1/2} = \sqrt n /2$. Curiously
enough, we were unable to find a similar result in the literature
while an upper bound on this approximation error in $O(1/n)$ when $n$
increases is well known (see \eg \cite{blyth1980expected,diaconis}).
Specifically, we want to prove that
$$
\textstyle | M_{n} - \gone \sigma_n |\ \geq\ C \sigma_n \, n^{-1},
$$
for some absolute constant $C>0$ and all $n \geq 1$.

We start from the Stirling's approximation of the factorial with an error
bound due to {R.~W.~Gosper}~\cite{gosper} and redeveloped more clearly
in
\cite{jveeh_stirling} (see also \cite{mortici11} for a similar bound):
\begin{equation}
  \label{eq:stirling-approx}
\textstyle n^n e^{-n} \sqrt{2\pi(n+\tinv{6})}\ \leq\ n!\ \leq\ n^n e^{-n} \sqrt{2\pi(n+\tinv{5})}.
\end{equation}
However, De Moivre gave the following exact formula for $M_{2n}$ \cite{diaconis},
$$
\textstyle M_{2n} := n\, 2^{-2n}\, {2n \choose n} = n \,2^{-2n} \,\tfrac{(2n)!}{(n!)^2}.
$$ 
Therefore, applying \eqref{eq:stirling-approx}
on this formula and using $\sqrt{1+ x} \leq 1 + \tinv{2} x$ for $x\geq -1$, we find for $n \geq 1$  
\begin{align*}
M_{2n}&\textstyle \leq n \,2^{-2n} \,\tfrac{(2n)^{2n} e^{-2n}
  \sqrt{2\pi(2n+\inv{5})}}{n^{2n} e^{-2n} 2\pi(n+\inv{6})}\ =\ \,\tfrac{n
  \sqrt{n+\inv{10}}}{\sqrt{\pi}(n+\inv{6})}\ <\ \gone \sigma_{2n} \sqrt{\tfrac{
  n}{n+\inv{6}}}\\\\
&\textstyle = \gone \sigma_{2n} \sqrt{1 - \tfrac{1
  }{6n+1}} \leq \gone \sigma_{2n}\, (1 - \tfrac{1
  }{12n+2}) \leq \gone \sigma_{2n}\, (1 - \tfrac{1
  }{14n}),
\end{align*}
or equivalently 
$$
\gone\sigma_{2n} - M_{2n} \geq C\,\sigma_{2n}\,(2n)^{-1},
$$
with $C = 1/7$, which provides the result.

\end{document}